\documentclass[twocolumn,10pt]{IEEEtran}
\usepackage{braket}
\newcommand{\X}{{\scriptscriptstyle\mathcal{X}}}

\input{def.tex}
\usepackage{dsfont}
\usepackage{minibox}
\DeclareSymbolFont{matha}{OML}{txmi}{m}{it}
\DeclareMathSymbol{\varv}{\mathord}{matha}{118}
\usepackage{eurosym}
\usepackage{hhline,physics}

\usepackage{multicol}
\usepackage{algorithm}
\usepackage{graphicx}
\usepackage{textcomp}
\usepackage{caption,pifont}
\usepackage[multiple]{footmisc}

\usepackage{algorithmic}

\newcommand{\trans}{^{\mathsf{T}}}
\newcommand{\herm}{^{\H}}
\makeatletter

\IEEEoverridecommandlockouts
\begin{document}
	\title{Max-Min SINR Analysis of STAR-RIS Assisted Massive MIMO Systems with Hardware Impairments} 
	\author{Anastasios Papazafeiropoulos,  Pandelis Kourtessis, Symeon Chatzinotas \thanks{A. Papazafeiropoulos is with the Communications and Intelligent Systems Research Group, University of Hertfordshire, Hatfield AL10 9AB, U. K., and with SnT at the University of Luxembourg, Luxembourg.  P. Kourtessis is with the Communications and Intelligent Systems Research Group, University of Hertfordshire, Hatfield AL10 9AB, U. K. S. Chatzinotas is with the SnT at the University of Luxembourg, Luxembourg. A. Papazafeiropoulos was supported  by the University of Hertfordshire's 5-year Vice Chancellor's Research Fellowship. 		S. Chatzinotas   is supported by the National Research Fund, Luxembourg, under the project RISOTTI. E-mails: tapapazaf@gmail.com,  p.kourtessis@herts.ac.uk, symeon.chatzinotas@uni.lu.}}
	\maketitle\vspace{-1.7cm}
	\begin{abstract}
		Reconfigurable intelligent surface  (RIS) has emerged as a cost-effective solution to improve wireless communication performance through just passive reflection. Recently, the concept of simultaneously transmitting and reflecting RIS (STAR-RIS) has appeared but the study of minimum signal-to-interference-plus-noise ratio (SINR) and the impact of hardware impairments (HWIs) remain  open. In addition to previous works on STAR-RIS, we consider a massive multiple-input multiple-output (mMIMO) base station (BS) serving multiple user equipments (UEs) at both sides of the RIS. Specifically, in this work, focusing on the downlink of a single cell,  we derive the minimum SINR obtained by the optimal linear precoder (OLP) with HWIs in closed form. The OLP maximises the minimum SINR subject to a given power constraint for any given passive beamforming matrix (PBM). Next, we obtain deterministic equivalents (DEs) for the OLP and the minimum SINR, which are then used to optimise the PBM.  Notably, based on the DEs and statistical channel state information (CSI), we optimise simultaneously the amplitude and phase shift  by using a projected gradient ascent algorithm (PGAM) for both energy splitting (ES) and mode switching (MS) STAR-RIS operation protocols with reduced feedback, \textcolor{black}{which is quite crucial for STAR-RIS systems that include the double number or variables compared to reflecting only RIS.} Simulations verify the analytical results, shed light on the impact of HWIs, and demonstrate the better performance of STAR-RIS compared to conventional RIS. \textcolor{black}{Also, a benchmark full instantaneous CSI  (I-CSI) based design is provided and shown to result in higher SINR but lower net achievable sum-rate than the statistical CSI based design because of large overhead associated with the acquisition of full I-CSI acquisition. Thus, not only do we evaluate the impact of HWIs but we also propose a statistical CSI based design that provides higher net sum-rate with low overhead and complexity.}
	\end{abstract}
	
	\section{Introduction}
	Reconfigurable intelligent surface (RIS) has appeared recently as an innovative technology that materialises a smart and reconfigurable radio environment through passive signal reflection by exploiting the advances in digitally-controlled metasurfaces \cite{Wu2019,Basar2019}. The structure of RIS includes a large number of nearly passive elements that reflect impinging waves independently by adjusting the amplitude and phase shift in real time. A RIS does not just adapt to the random and time-varying characteristics of the wireless channel by using traditional techniques but it can  tune its reflecting elements towards shaping the wireless propagation channel for enhanced signal transmission. Also,  the RIS  technology is accompanied by low hardware cost and power consumption demands because the signal reflection is passive and does  not consist of  any radio frequency (RF) chains while its lightweight structure promotes its flexible deployment.
	
	Due to these significant advantages, RIS has been investigated in various wireless scenarios. In particular, in \cite{Wu2019}, the joint active and passive beamforming for both single UE and multiple  UEs RIS-assisted multiple-input single-output (MISO) systems were first formulated while minimizing the total transmit power at the base station (BS). Pan et al.  studied the weighted sum-rate maximisation problems in multi-cell multiple-input multiple-output (MIMO) communications \cite{Pan2020},  the artificial noise added secure MIMO communications  \cite{hong2020artificial}, etc., which all demonstrate the significant performance gains  due to a RIS deployment. However, all these works rely on instantaneous channel state information   (I-CSI ) \cite{Wu2019,Pan2020,hong2020artificial}. 
	
	Generally, there are two approaches used for phase shift optimization namely I-CSI and statistical CSI (S-CSI) \cite{Zhao2020,Kammoun2020,Papazafeiropoulos2021,VanChien2021,Papazafeiropoulos2021a,Chen2021,You2021,Zhi2022,Zhang2022}. Contrary to the former approach, depending on small-scale fading and varying at each coherence interval, the latter depends on large-scale statistics in terms of path-losses and correlation matrices that vary at every several coherence intervals, and thus, come with reduced overhead. Note that the overhead amount is crucial because it increases with the number of RIS elements. Especially, in high mobility setups, the acquisition of I-CSI can be very challenging since the RIS has to be optimised very frequently. Hence, the choice of the S-CSI approach appears to be more beneficial in realistic scenarios.
	
	In this direction, the appealing characteristics of RIS have spurred a growing interest in both academia and industry \cite{Basar2019,Wu2019,Pan2020,Bjoernson2019b,Kammoun2020, Elbir2020,Guo2020,Chen2019,Papazafeiropoulos2022b,Papazafeiropoulos2022a}, while different setups  have been studied such as  multi-antenna \cite{Zhang2020a} and multi-cell communications \cite{Pan2020}, double  RIS \cite{Zheng2021,Papazafeiropoulos2022c} or multiple RIS networks \cite{Mei2022,Papazafeiropoulos2021a}. However, the majority of existing works have relied on the scenario where only reflection occurs, i.e., both the transmitter and the receiver are located on the same side of the surface, while, in practice, user equipments (UEs) can be found on both sides of the surface, i.e., in front and behind the RIS. Fortunately, further advancements in programmable metamaterials have resulted in the achievement of simultaneously transmitting and reflecting RIS (STAR-RIS) that satisfy these application requirements \cite{Xu2021,Mu2021,Niu2021,Wu2021,Niu2022}. In fact, STAR-RIS	has been proposed to provide full space $ (360^{\circ}) $ coverage by tuning both the amplitudes and phase shifts of the impinging waves. For example, in \cite{Xu2021}, a general hardware model and two-channel models that correspond to the near-field and far-field regions of STAR-RIS have been provided in the case of only two UEs. Notably,  therein, it has been shown that the coverage and the diversity are  better than conventional RIS-assisted systems with reflecting only elements. Moreover, in \cite{Mu2021}, three operation protocols were suggested to adjust the transmission reflection coefficients of the transmission signals. Namely, these protocols are energy splitting (ES), mode switching (MS), and time switching (TS). \textcolor{black}{However, most of these works rely on I-CSI, which is an unattractive approach for STAR-RIS systems  that include a double number of variables with respect to reflecting-only RIS.} Also, there is no work that has  obtained the optimal linear precoder (OLP) and the minimum SINR and their deterministic equivalent (DE) expressions for STAR-RIS.

	In the meanwhile, although independent Rayleigh fading is an unrealistic assumption not only many initial works but also ongoing studies on conventional RIS have not taken RIS correlation into account e.g.,  \cite{Basar2019,Wu2019,Pan2020}, while in \cite{Bjoernson2020}, it was shown that the spatial correlation has to be considered. To address this issue, recently, several works on RIS  have taken correlation into consideration such as \cite{Kammoun2020,Nadeem2020,Papazafeiropoulos2021}, but only two works have studied its effect on STAR-RIS assisted systems \cite{Mu2021,Niu2021}. In parallel, regarding STAR-RIS, most works have assumed a single-antenna transmitter.

	On this ground, it is worthwhile to mention that a  part of prior literature on RIS-assisted systems \cite{Li2020,Xing2020,Qian2020,Liu2020,Shen2020,Zhou2020a} has considered the  inevitable impact of transceiver hardware impairments (T-HWIs) \cite{Papazafeiropoulos2017a} such as the quantization noise in the analog-to-digital converters (ADCs) and the  in-phase/quadrature-phase (I/Q)-imbalance~\cite{Qi2010}, which remain after the application of compensation/mitigation algorithms \cite{Schenk2008,Bjoernson2017}. Moreover,  HWI at the RIS (RIS-HWI), known as phase noise and appearing due to finite precision configuration of the phase shifts, has been investigated in \cite{Badiu2019,Xing2020,Qian2020}. Notably, in \cite{Papazafeiropoulos2021} and \cite{Papazafeiropoulos2021b}, the impact of T-HWIs and RIS-HWI on the sum achievable rate and minimum rate have been studied, respectively. However, there is no work on STAR-RIS that has studied the impact of HWIs.
	\subsection{Contribution}	 
	\textcolor{black}{	The previous observations motivate the topic of this work, which is the design of realistic STAR-RIS-assisted systems with T-HWIs and RIS-HWI under correlated Rayleigh fading conditions. \textcolor{black}{Notably, the introduction of STAR-RIS imposes new challenges. In particular, the first constraint is not simple but includes the two types of passive beamforming, namely transmission and reflection beamforming, to be optimized, which are coupled with each other due to the energy conservation law.} Also, the study of HWIs increases the complexity and the difficulty during the mathematical manipulations in the following analysis of STAR-RIS assisted systems. The main contributions are summarised as follows:}
	\begin{itemize}
		\item \textcolor{black}{We aim at characterising the potentials of STAR-RIS assisted systems with multiple UEs at the same time-frequency response under the realistic conditions of correlated fading and HWIs. Under this scenario, we derive the minimum SINR and obtain the  OLP with HWIs in closed forms.  Next, we obtain the DE expressions of the OLP and minimum SINR at the large system limit in terms of large-scale statistics.}
		
		\item\textcolor{black}{ Contrary to \cite{Kammoun2020}, we have assumed a STAR-RIS system and we have included it in the design of the HWIs. Also, compared to other works on STAR-RIS, this is the only word studying not only the minimum SINR but also  the effects of HWIs. Moreover, we have provided an analytical framework to present in an elegant unified way the analysis of both types of UEs located in front of and the behind the surface. We have considered multiples UEs at each side of the surface, while most works on STAR-RIS have assumed only one UE at each side. In addition, note that many previous works have assumed that the channel between the BS and the RIS is deterministic expressing a line-of-sight (LoS) component, e.g., \cite{Kammoun2020}, while the analysis here is more general since we assume that all links are correlated Rayleigh fading distributed.}
		
		\item  \textcolor{black}{We optimise the minimum SINR with respect to both the amplitude and phase shifts. Specifically, despite its non-convexity, we obtain an iterative efficient method that enables their simultaneous at each iteration.   Note that other works have optimised the amplitude and the phase shifts in  an alternating way. \textcolor{black}{Our proposed optimization  is quite  advantageous for STAR-RIS implementations, which require double optimization variables compared to reflecting only RIS.}
			\item 	Simulations are depicted against the analytical results for verification. We show the outperformance  over conventional RIS, and we investigate the impact of various parameters such as the HWIs. For instance, T-HWIs, being power dependent, have a more severe impact at high transmit power. \textcolor{black}{Also, we consider as a performance benchmark a full  I-CSI based RIS design, and we depict the superiority of the proposed framework based on S-CSI. Furthermore, we compare the gradient ascent with a genetic algorithm (GA), and we show that the former performs better.}}
	\end{itemize}
	
	\subsection{Paper Outline} 
	The remainder of this paper is organized as follows. Section~\ref{System} presents the system model of a STAR RIS-assisted mMIMO system with correlated Rayleigh fading.  Section~\ref{PerformanceAnalysis} presents the downlink minimum rate with  correlated fading  and HWIs. Section~\ref{optimization} provides the optimization by deriving the OLP and passive beamforming matrix (PBM) of the STAR-RIS after having obtained the DE of the minimum rate. 
	The numerical results are provided in Section~\ref{Numerical}, and Section~\ref{Conclusion} concludes the paper.
	
	\subsection{Notation}Vectors and matrices are denoted by boldface lower and upper case symbols, respectively. The notations $(\cdot)^\T$, $(\cdot)^\H$, and $\tr\!\left( {\cdot} \right)$ describe the transpose, Hermitian transpose, and trace operators, respectively. The expectation operator is described by $\EE\left[\cdot\right]$ while $ \diag\left(\ba \right) $ describes an $ n\times n $ diagonal matrix with diagonal elements being the elements of vector $ \ba $. In the case of a matrix $ \bA $, $\diag\left(\bA\right) $ represents a diagonal matrix with elements corresponding to the diagonal elements of $ \bA $. Also, $ \arg\left(\cdot\right) $  denotes the argument function. Finally, $\bb \sim \cC\cN{(\b0,\mathbf{\Sigma})}$ expresses a circularly symmetric complex Gaussian vector with {zero mean} and covariance matrix $\mathbf{\Sigma}$.
	\section{System Model}\label{System}
	As shown in Fig~\ref{Fig1}, we consider a STAR-RIS-aided system, where a BS deployed with an $ M $-element uniform linear array (ULA) serves simultaneously a set of $ \mathcal{K}=\{1, \ldots, K\} $ single-antenna UEs that can be located in any of both sides of the STAR-RIS. In particular, a set of $ \mathcal{K_{\mathrm{t}}}=\{1, \ldots, K_{\mathrm{t}}\} $ UEs are found in the transmission region $ (t) $ and a set of $ \mathcal{K_{\mathrm{r}}}=\{1, \ldots, K_{\mathrm{r}}\} $ UEs are found in the reflection region (\textcolor{black}{$ r $}), respectively. In other words, we have $ K_{\mathrm{t}}+K_{\mathrm{r}} =K $ UEs. \textcolor{black}{We can define the RIS operation mode for each of the $ $ K $ $ UEs in a unified way elegant way by defining the set $ \mathcal{W}=\{w_{1}, \ldots, w_{K}\} $. If $ w_{k}=t $, it means that the $ k $th UE is found behind the STAR-RIS, i.e., it is located in the transmission region, and we have $ k \in  \mathcal{K_{\mathrm{t}}} $. Otherwise, if  $ w_{k}=r $, the $ k $th UE is found at the same side of the RIS with the BS, i.e., it is located in the reflection region, which means $ k \in  \mathcal{K_{\mathrm{r}}} $.} Also, we assume that blockages do not allow the presence of direct links between the BS and the UEs, which justifies the application of a RIS. Of course, the following analysis can be extended to include any direct links. Similarly, signals that are impinging on the RIS two or more times are ignored due to negligible power.
	
	\begin{figure}[!h]
		\begin{center}
			\includegraphics[width=0.8\linewidth]{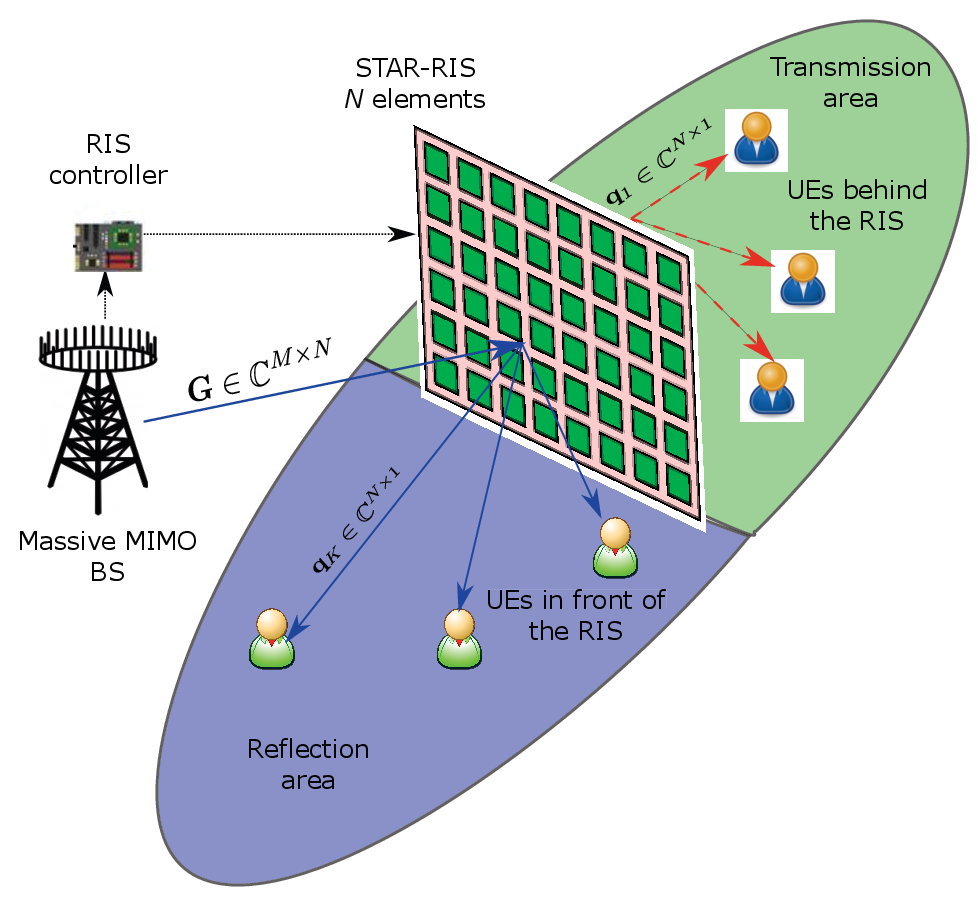}
			\caption{\footnotesize{ A mMIMO STAR-RIS assisted system with multiple UEs at transmission and reflection regions.  }}
			\label{Fig1}
		\end{center}
	\end{figure} 
	
	Regarding the STAR-RIS, which is at the center of our analysis, we assume that it is deployed by a uniform planar array (UPA), consisting of a total of $ N=N_{\mathrm{h}}\times N_{\mathrm{v}} $ elements \textcolor{black}{belonging to the set $ \mathcal{N}=\{1, \ldots, N\}  $}, where $ N_{\mathrm{h}} $ and  $ N_{\mathrm{v}} $ are the numbers of horizontally and  vertically passive elements, respectively. According to its principle, the STAR-RIS is able to configure the transmitted $ (t) $ and reflected $ (r) $ signals by using two independent coefficients. Specifically, we denote $ t_{n} =( {\beta_{n}^{t}}e^{j \phi_{n}^{t}})s_{n}$ and $ r_{n}=( {\beta_{n}^{r}}e^{j \phi_{n}^{r}})s_{n} $  the transmitted 	and reflected signal by the $ n $th STAR-RIS element, respectively, \textcolor{black}{where $ s_{n} $ is the impinging signal}. Herein, the amplitude and phase parameters  $ {\beta_{n}^{w_{k}}}\in [0,1] $ and $ \phi_{n}^{w_{k}} \in [0,2\pi)$ are independent, and the $ k $th UE can be in any of the two regions of the RIS \cite{Xu2021}. Based on this model, the choice of the amplitudes is based on the relationship expressed by the law of energy conservation as
	\begin{align}
		(\beta_{n}^{t})^{2}+(\beta_{n}^{r})^{2}=1,  \forall n \in \mathcal{N},
	\end{align}
	but the phases  $ \phi_{n}^{t} $ and $ \phi_{n}^{r} $ can be chosen independently.
	%

	\subsection{STAR-RIS Protocols}
	STAR-RIS has been presented mainly by means of its ES/MS protocols  \cite{Mu2021}. Below, we present them briefly in terms  of their main points.
	\subsubsection{ES protocol} All STAR-RIS elements serve simultaneously all UEs, which are located  in both   $ t $ and $ r $ regions. In this case, the  PBM for the $ k
	$th  UE is described as $ \bPhi_{w_{k}}^{\mathrm{ES}}=\diag( {\beta_{1}^{w_{k}}}\theta_{1}^{w_{k}}, \ldots,  {\beta_{N}^{w_{k}}}\theta_{N}^{w_{k}}) \in \mathbb{C}^{N\times N}$, where $ \beta_{n}^{w_{k}} \ge 0 $, $ 		(\beta_{n}^{t})^{2}+(\beta_{n}^{r})^{2}=1 $, and $ |\theta_{i}^{w_{k}}|=1, \forall n =1, \dotsb, N $.
	
	\subsubsection{MS protocol} The STAR-RIS is divided into two groups of $ N_{t} $ and $ N_{r} $ elements that serve UEs in the  $  t$ and $ r $ regions, respectively. Obviously, we have  $ N_{t}+N_{r}=N $. The PBM for UE $  k \in \mathcal{K}_{t} $ or $  k \in \mathcal{K}_{r} $
	is given by $ \bPhi_{w_{k}}^{\mathrm{MS}}=\diag( {\beta_{1}^{w_{k}}} \theta_{1}^{w_{k}}, \ldots,  {\beta_{N}^{w_k}}\theta_{N}^{w_{k}}) \in \mathbb{C}^{N\times N}$, where $ \beta_{n}^{w_{k}}\in \{0,1\}$, $ 	(\beta_{n}^{t})^{2}+(\beta_{n}^{r})^{2}=1 $, and $ |\theta_{i}^{w_{k}}|=1, \forall n \in \mathcal{N} $. 
	\begin{remark}
		The  amplitude coefficients for transmission and reflection are restricted to binary values, which means that the MS protocol is  a special case of the ES protocol. Consequently, the ES protocol is superior to the MS counterpart since the latter  cannot achieve the full-dimension transmission and reflection beamforming gain. However, a trade-off appears since the MS protocol comes with the advantage of lower computational complexity regarding the  PBM design.
	\end{remark}

	\subsection{Channel model}
	Under these conditions, we consider narrowband quasi-static block fading channels based on a time-division duplex protocol that is preferable in mMIMO systems. Let 
	\begin{align}
		\bG&=\sqrt{\tilde{ \beta}}\bR_{\mathrm{BS}}^{1/2}\bD\bR_{\mathrm{RIS}}^{1/2},\label{eq2}\\
		\bq_{k}&=\sqrt{\tilde{ \beta}_{k}}\bR_{\mathrm{RIS}}^{1/2}\bc_{k}
	\end{align}
	be the  channels under correlated Rayleigh fading conditions between the BS and the STAR-RIS, and between  the STAR-RIS and UE $ k $, respectively. The correlation matrices $ \bR_{\mathrm{BS}} \in \mathbb{C}^{M \times M} $ and $ \bR_{\mathrm{RIS}} \in \mathbb{C}^{N \times N} $ correspond to the  BS and the RIS, and are  deterministic Hermitian-symmetric positive semi-definite. The former, i.e., $  \bR_{\mathrm{BS}} $  can be modeled e.g., as in \cite{Hoydis2013}, and the latter, which is $ \bR_{\mathrm{RIS}} $ is modeled as in \cite{Bjoernson2020}.  \textcolor{black}{Specifically, we assume that $  \bR_{\mathrm{BS}}=[\bA~\b0_{M\times M-P}]$, where $ \bA=[\ba(\varphi_{1}) \cdots \ba(\varphi_{P})] \in \mathbb{C}^{N\times P}$ is composed of the steering vector $ \ba(\varphi)\in \mathbb{C}^{N} $ defined as
		\begin{align}
			\ba(\varphi)=\frac{1}{\sqrt{P}}[1, e^{-i2\pi\omega \sin(\varphi), \ldots,-i2\pi\omega (N-1)\sin(\varphi)}]^{\T}, 
		\end{align}
		where $ \omega =0.3 $ is the antenna spacing in multiples of the wavelength and $ \varphi_{p}=-\pi/2+(p-1)\pi/P $, $ p= 1, \ldots, P$ are the uniformly distributed angles of transmission with $ P=M/2 $ \cite{Hoydis2013}. Similarly, we describe $  \bR_{\mathrm{R}} $.  Also, the $ (l,m) $th entry of $ \bR_{\mathrm{RIS}} $ is described as
		\begin{align}
			[\bR_{\mathrm{RIS}}]_{l,m}=\mathrm{sinc}\Big(\frac{2\|\bu_{l}-\bu_{m}\|}{\lambda}\Big), ~~\forall \{l,m\}\in \mathcal{N},
		\end{align}
		where $ \|\bu_{l}-\bu_{m}\| $ expresses the distance between the $ l $th and $ m $th  STAR-RIS elements, and $ \lambda $ is the wavelength \cite{Bjoernson2020}.} The vectors $ \mathrm{vec}(\bD)\sim \mathcal{CN}\left(\b0,\Id_{MN}\right) $ and  $ \bc_{k} \sim \mathcal{CN}\left(\b0,\Id_{N}\right) $ describe the corresponding fast-fading components. Also, $\tilde{ \beta} $ and $\tilde{ \beta}_{k} $ describe the  path-losses  of the BS-RIS and RIS-UE $ k $ links in $ t $ or $r  $ region, respectively. Together with the correlation matrices, they are assumed to be known by the network.

	\section{Downlink Data Transmission}\label{PerformanceAnalysis}
	We present the ideal downlink signal model, introduce the hardware impairments at both the STAR-RIS and the transceiver, and continue with the realistic model including the HWIs. The downlink transmission of data from the BS to all UEs consists of a broadcast channel, which requires a specific precoding strategy in terms of a precoding vector $\bff_{k} \in \bbC^{M \times 1}$. 
	
	\subsection{Ideal  Signal Model}\label{SignalModel}
	The ideal received complex baseband signal  by UE $ k $ is written as
	\begin{align}
		y_{k}=\bar{\bh}^\H_{k}\bs+z_{k},\label{eq:Ypt1}
	\end{align}
	where    $\bs= \sum_{i=1}^{K}\sqrt{p_{i}}\bff_{i}l_{i}$ denotes the transmit signal vector  by the BS with $ \bff_{i} \in \mathbb{C}^{M \times 1}$,  $ p_{i} $, and $ l_{i}\sim \mathcal{CN}(0,1) $ being the precoder, the transmit power, and data symbol for UE $ i $, respectively.  Moreover, $ \bar{\bh}_{k}=  \bG\bPhi_{w_{k}}^{\mathrm{\X}}	\bq_{k} $ is the overall channel vector for $ \mathcal{X}=\{ES,MS\} $. Also, $z_{k} \sim \cC\cN(0,\sigma^{2})$ is the additive  white complex Gaussian noise at UE $k$. The transmit signal satisfies the following average power constraint per UE 
	\begin{align}
		\EE[\|\bs\|^{2}]=\tr(\bP\bF^{\H}\bF)\le P_{\mathrm{max}},
	\end{align}
	where $ \bP=\diag(p_{1}, \ldots, p_{K}) $, $ \bF=[\bff_{1}, \ldots, \bff_{K}] \in\mathbb{C}^{M \times K}$, and $ P_{\mathrm{max}} $ is the transmit power constraint at the BS.

	\subsection{RIS Phase Noise}
	In practice, phase errors appear because it is not possible to configure the RIS elements with infinite precision  \cite{Badiu2019}. Mathematically, a random diagonal phase error matrix consisting of $ N $ random phase errors can be used to model this phase noise for  each element. Specifically, we denote $ \tilde{\bPhi}_{w_{k}}^{\X}=\diag(\tilde{\theta}_{1}^{w_{k}}, \ldots,  \tilde{\theta}_{N}^{w_{k}}) \in \mathbb{C}^{N\times N}$, where $ \mathcal{X}=\{\mathrm{ES},\mathrm{MS}\} $ with $ \tilde{\theta}_{i}^{w_{k}}, i=1,\ldots,N $ being the random phase errors of the RIS phase shifts of UE $ k $ that are i.i.d. randomly distributed in $ [-\pi, \pi) $ and based on a certain circular distribution. The probability density function (PDF) of $\tilde{\theta}_{i}^{w_{k}} $ is assumed symmetric with its mean direction equal to zero, i.e., $ \arg\left(\EE[\mathrm{e}^{j \tilde{\theta}_{i}^{w_{k}}}]\right)=0 $ \cite{Badiu2019}. 
	
	The most common  PDFs that could be used to describe the phase noise on a RIS are the Uniform and the Von Mises distributions \cite{Badiu2019}. The former describes a complete lack of knowledge and its  characteristic function (CF), denoted by $ m $, is equal to $0 $. The latter has a zero mean and concentration parameter $ \kappa_{\tilde{\theta}} $, 	which captures the accuracy of the estimation. Its CF is $ m= \frac{\mathrm{I}_{1}\!\left(\kappa_{\tilde{\theta}}\right)}{\mathrm{I}_{0}\!\left(\kappa_{\tilde{\theta}}\right)}$, where $ \mathrm{I}_{p}\!\left(\kappa_{\tilde{\theta}}\right)$ is the modified Bessel function of the first kind and order 	$ p $.

	In the following analysis, we are going to use the variance $ \bar{\bR}_{k}=\EE\{\bar{\bh}_{k}\bar{\bh}_{k}^{\H}\} $ of the cascaded channel vector for UE $ k $ $ \bar{\bh}_{k}= \bG\bPhi_{w_{k}}^{\X}\tilde{\bPhi}_{w_{k}}^{\X} \bq_{k} $ given the PBM. In particular, $ \bar{\bR}_{k} $ can be written as
	\begin{align}
		\bar{\bR}_{k}&=\EE\{\bar{\bh}_{k}\bar{\bh}_{k}^{\H}\}\\
		&=\EE\{ \bG\bPhi_{w_{k}}^{\X}\tilde{\bPhi}_{w_{k}}^{\X} \bq_{k}  \bq_{k}^{\H}(\tilde{\bPhi}_{w_{k}}^{\X} )^{\H} (\bPhi_{w_{k}}^{\X})^{\H}\bG^{\H}\}\label{eq7}\\
		&=\tilde{ \beta}_{k}\EE\{ \bG\bPhi_{w_{k}}^{\X}\tilde{\bPhi}_{w_{k}}^{\X} \bR_{\mathrm{RIS}}(\tilde{\bPhi}_{w_{k}}^{\X} )^{\H} (\bPhi_{w_{k}}^{\X})^{\H}\bG^{\H}\}\label{eq8}\\
		&=
		\hat{\beta}_{k}\tr(\bR_{\mathrm{RIS}}\bPhi_{w_{k}}^{\X} \tilde{\bR}_{\mathrm{RIS}}  (\bPhi_{w_{k}}^{\X})^{\H})\tilde{\bR}_{\mathrm{BS}},\label{cov1}
	\end{align}
	where in \eqref{eq7}, we  have  substituted the cascaded  channel. In \eqref{eq8}, we have exploited the independence between $ \bG $ and $ \bq_{k} $, and have considered that  $ \EE\{	\bq_{k}	\bq_{k}^{\H}\} =\tilde{ \beta}_{k} \bR_{\mathrm{RIS}}$.   Next, we have used $ \tilde{\bR}_{\mathrm{IRS},k}=\EE\{\tilde{\bPhi}_{w_{k}}^{\X} \bR_{\mathrm{RIS}}(\tilde{\bPhi}_{w_{k}}^{\X} )^{\H}\}= m^{2}\bR_{\mathrm{IRS},k}+\left(1-m^{2}\right)\Id_{N}$ with $ m $ denoting the CF of phase noise \cite[Eq. 12]{Papazafeiropoulos2021}. Moreover, in \eqref{cov1}, we used that   $ \EE\{\bD \bU\bD^{\H}\} =\tr (\bU) \Id_{M}$ with $\bU  $ being a deterministic square matrix. Also, we have denoted that $\hat{\beta}_{k}= \tilde{ \beta}\tilde{ \beta}_{k} $. Henceforth, we omit the superscript $ \mathcal{X} $ for the sake of clarity. However, in Section \ref{Numerical}, we are to examine  both protocols.
	\begin{remark}
		In the case of independent Rayleigh fading, i.e., when   $\bR_{\mathrm{RIS}} =\bR_{\mathrm{BS}} =\Id_{N} $, the variance of the cascaded channel is $ \bar{\bR}_{k} =\hat{\beta}_{k}( \sum_i^{N}(\beta_{i}^{w_{k}})^{2}) \Id_{M}$, i.e., it  does not depend on the phase shifts but only on the amplitudes \cite{Papazafeiropoulos2022}. Hence, $ \bar{\bR}_{k} $ cannot be optimized with respect to the phase shifts but only regarding the amplitudes. In the special case that we have optimal reflection and transmission, which means $ \beta_{i}^{w_{k}}=1 $, we obtain $ \bar{\bR}_{k} =\hat{\beta}_{k}N\Id_{M}$ that increases with the number of RIS elements.
	\end{remark} 
	
	\subsection{Transceiver Hardware Impairments}
	In practical systems, both the transmitter and the receiver are impaired by hardware distortions. Generally, we meet  them at both the transmitter and the receiver. Herein, we focus on the additive power-dependent distortion noises at the transmitter and receiver
	while  the study of multiplicative phase noise at the transceiver oscillators is left for future work. Note that given that we focus on STAR-RIS, the study of its impairment is more important than any transceiver distortions, which have been already studied in the literature. However, these impairments have not been taken into account during the analysis of STAR-RIS-assisted systems. It is worthwhile to mention that  manufacturers give certain hardware specifications, which allow assuming that   the HWIs parameters are known.

	In particular, additive distortion noise appears at both the transmitter and the receiver because of inevitable residual  impairments coming from imperfect compensation of the quantization noise in the Analog-to-Digital Converters (ADCs) at the receiver, the I/Q
	imbalance, etc. \cite{Studer2010}. The consequence is that a mismatch appears between the signal that is intended to be transmitted and the generated signal at the transmitter side, while the  signal is distorted during the reception processing at the receiver side.

	Mathematically speaking the transmitter and receiver distortion noises are distributed as complex Gaussian variables, which is corroborated by measurement results \cite{Studer2010,Wenk2010}. The justification of the Gaussianity relies on the fact that the additive distortion can be seen as the result of the aggregate contribution of many impairments. Specifically, if $ \bs \in \mathbb{C}^{M \times 1} $ is the data signal vector, we have
	\begin{align}
		\etav_{\mathrm{BS}}&\sim \mathcal{CN}(\b0, \bm \Upsilon^{\mathrm{BS}}),\label{HWI_BS}\\
		\eta_{\mathrm{UE},k}&\sim\mathcal{CN} (0,\upsilon_{k}^{\mathrm{UE}}),\label{HWI_UE}
	\end{align}
	where $ \bm \Upsilon^{\mathrm{BS}}=\kappa_{\mathrm{BS}}^{2}\diag(W_{11},\ldots, W_{MM}) $ and $ \upsilon_{k}^{\mathrm{UE}}=\kappa_{\mathrm{UE}}^{2} \bh_{k}^{\H}\bW \bh_{k} $ with $ \bW=\EE\{\bs\bs^{\H}\}=\sum_{i=1}^{K}{p_{i}} $, $ W_{ii} $ is the $ i $th diagonal element of $ \bW $ and $ P_{\mathrm{max}}=\tr(\bW) $. Without any loss, we assume that all UEs present the 	same level of impairments. The variables $ \kappa_{\mathrm{BS}}^{2} $ and $ \kappa_{\mathrm{UE}}^{2} $ are proportionality coefficients, which denote the severity of T-HWIs. 
	
	\subsection{Realistic  Signal Model with HWIs}\label{SignalModel1}
	The received signal by UE $ k $, impaired by both RIS phase noise and T-HWIs, is given by 
	\begin{align}
		y_{k}=\bar{\bh}^\H_{k}(\bs+\etav_{\mathrm{BS}})+\\eta_{\mathrm{UE},k}+z_{k}.\label{DLreceivedSignal}
	\end{align}
	The additive distortions in \eqref{HWI_BS} and \eqref{HWI_UE} can be written as
	\begin{align}
		\Upsilon^{\mathrm{BS}}&=\kappa_{\mathrm{BS}}^{2}\sum_{i=1}^{K}{p_{i}}\Id_{M},\\
		\upsilon^{\mathrm{UE}}&=\sum_{i=1}^{K}{p_{i}}\kappa_{\mathrm{UE}}^{2}\bh_{k}^{\H} \bh_{k}.
	\end{align}

	Thus, under perfect CSI conditions, the downlink SINR at $  k$th  UE is given by \eqref{SINR10}, 
\begin{figure*}
	\begin{align}
		\gamma_{k}=\frac{{p_{k}}|\bh_{k}^{\H}\bff_{k}|^{2}}{\sum_{i\ne k}{p_{i}}|\bh_{k}^{\H}\bff_{i}|^{2}+\kappa_{\mathrm{BS}}^{2}|\sum_{i=1}^{K}{p_{i}}\bh_{k}^{\H} \bh_{k}|^{2}+\sum_{i=1}^{K}{p_{i}}\kappa_{\mathrm{UE}}^{2}|\bh_{k}^{\H} \bh_{k}|^{2} +1}\label{SINR10}.
	\end{align}
\line(1,0){490}
\end{figure*}
	where $ \bh_{k}\sim \mathcal{CN}(\b0, \bR_{k}) $ with $ \bR_{k}=	\bar{\bR}_{k}/\sigma^{2}=\rho_{k}\tr(\bR_{\mathrm{RIS}} \bPhi_{w_{k}} \bR_{\mathrm{RIS}}  \bPhi_{w_{k}}^{\H})\bR_{\mathrm{BS}}$ by denoting $ \rho_{k}=\hat{\beta}_{k}/\sigma^{2} $ to normalize the noise variance to unity.

	\section{Minimum SINR Maximization}\label{optimization}
	The focus of this work is the optimization of the downlink minimum  SINR under constraints regarding the   sum-power and the PBM.
	The optimization problem is formulated as
	\begin{align}
		\!\!	(\mathcal{P}1)~&\max_{  \bp, \{\thetv,\betv\}} \min_{k} \frac{\gamma_{k}\left( \bp, \{\thetv,\betv\}\right)}{a_{k}}\label{Maximization1} \\
		&~\mathrm{s.t.}~~~~~~~~\bw^{\T}\bp \le P_{\mathrm{max}},~p_{k}>0,~ \|\bff_{k}\|=1, \forall k\label{Maximization3} \\
		&~~~~~~~~~~~~~ (\beta_{n}^{t})^{2}+(\beta_{n}^{r})^{2}=1,  \forall n \\
		&~~~~~~~~~~~~~ \beta_{n}^{t}\ge 0, \beta_{n}^{r}\ge 0,~\forall n \\
		&~~~~~~~~~~~~~ |\theta_{n}^{t}|=|\theta_{n}^{r}|=1, ~\forall n \label{Maximization4} 
	\end{align}
	where $ \ba_{l}=[a_{1}, \ldots, a_{K}]^{\T} $ is the priority vector with $ a_{k}>0 $ being the priority assigned to UE $ k $ by the network, $ \bw_{l}=[w_{1}, \ldots, w_{K}]^{\T} $ is the weight vector with $ w_{k} $  being associated with power $ p_{k} $ in the power constraint,  and $ \bp=[p_{1}, \ldots, p_{K}]^{\T} $. In the case that a specific UE has to be excluded from the power constraint, its corresponding weight can be chosen to be extremely small. 
	
	The problem in $ (\mathcal{P}1) $ is non-convex and presents coupling among the optimization variables, which are the active and passive beamforming at the BS and the STAR-RIS, respectively. This difficulty is tackled by following the common  alternating optimization in two stages. \textcolor{black}{In the first stage,  for any given passive beamforming, we provide the optimal linear precoder design in terms of the OLP and the optimal allocated power;  the next stage includes the STAR-RIS design. Specifically, in the first stage and  based on the nonlinear Perron–Frobenius theory, initially, we assume fixed transmit beamformers. Then, the optimization is reformulated as a geometric program, where we obtain the Lagrangian. Next,  we obtain the optimal transmit beamformers and the optimal downlink power. In the second stage, concerning the optimization of the STAR-RIS, we optimize simultaneously both the phase shifts and the amplitudes of each element.}

	\subsection{Optimal  Linear Precoder with HWIs}	
	For any given PBM, the following max-min SINR problem provides the OLP and the
	optimally allocated powers $ \bp $ by taking HWIs into account.\footnote{\textcolor{black}{Note that \cite{Cai2011} has obtained the OLP with no HWIs. However, the  extra term, corresponding to T-HWI,  induces many differences with increased difficultly in terms of manipulations with respect to the equations in \cite{Cai2011}.}}
	\begin{align}
		\!\!	(\mathcal{P}2)~&\max_{\bp,   \{\thetv,\betv\}} \min_{k} \frac{\gamma_{k}}{a_{k}}\left( \bp, \{\thetv,\betv\}\right)\label{Maximization10} \\
		&~\mathrm{s.t.}~~~~~~~~\bw^{\T}\bp \le P_{\mathrm{max}},~p_{k}>0,~ \|\bff_{k}\|=1, \forall k.\label{Maximization30} 
	\end{align}
	
	\begin{theorem}\label{TheoremOLP}
		The  solution to $ 	(\mathcal{P}2) $, i.e., the OLP with HWIs is given by
		\begin{align}
			\bff_{k}^{*}=\frac{\bSigma\bh_{k}}{\|\bSigma\bh_{k}\|}\label{precoder1},
		\end{align}
		where $ \bSigma= (\sum_{i \ne k}\xi_{i}^{*}			(\bar{\kappa}+1)\bh_{i}\bh_{i}^{\H}+\xi_{k}^{*}\bar{\kappa}\bh_{k}\bh_{k}^{\H}+w_{k}\Id_{M})^{-1}$ with $\bar{\kappa}= \kappa_{\mathrm{BS}}^{2}+\kappa_{\mathrm{UE}}^{2} $, and 
		the parameters $ \xi_{i} $ are the unique positive solution of the
		following fixed-point equations
		\begin{align}
			\xi_{k}^{*}=\frac{\delta^{*}}{\bh_{k}^{\H}\bSigma\bh_{k}}\label{xk}
		\end{align}
		with $ \delta^{*} $ being the minimum SINR under OLP that is provided by
		\begin{align}
			\delta^{*}=\frac{K P_{\mathrm{max}}}{\sum_{k=1}^{K}(\bh_{k}^{\H}\bSigma\bh_{k})^{-1}}.\label{tau}
		\end{align}
		
		Regarding the optimal power coefficients $ p_{i} $, they are obtained such that $ \frac{\gamma_{1}^{*}}{a_{1}}=\frac{\gamma_{2}^{*}}{a_{2}}=\cdots=\frac{\gamma_{K}^{*}}{a_{K}}=\delta^{*} $. 
		Thus, we obtain
		\begin{align}
			\bp^{*}=\delta^{*}(\Id_{K}-\delta^{*}\bD\bT)^{-1}\bD\one_{K},\label{dlpower}
		\end{align}
		where $ \bD=\diag(\frac{1}{|\bh_{1}^{\H}\bff_{1}^{*}|^{2}, \ldots, |\bh_{K}^{\H}\bff_{K}^{*}|^{2}}) $ , and $ [\bT]_{k,i}=|\bh_{k}^{\H}\bff_{i}| +b$ with $ b= (
		\kappa_{\mathrm{BS}}^{2}+\kappa_{\mathrm{UE}}^{2})|\bh_{k}^{\H} \bh_{k}|^{2}$, if $ i\ne k $ or $ b $ otherwise.

	\end{theorem}
	\begin{proof}
		Please see Appendix~\ref{TheoremOLPproof}.	
	\end{proof}
	
	In the case of ideal transceiver hardware and no phase noise at the RIS, we obtain the results in  \cite{Cai2011}.

	\subsection{Deterministic Equivalent Analysis}
	For the STAR-RIS design below, we resort to the large system analysis, where the various expressions become deterministic because: i) the matrix inversion operation in \eqref{tau}, which is highly computationally demanding, has to be computed at every iteration; ii) both expressions, i.e., $ \delta^{*} $ and $\xi_{k}^{*}  $ depend on instantaneous CSI, i.e., they vary at the order of milliseconds. In other words, they have to be computed at every channel realization. For these reasons, we take advantage of statistical CSI, which varies at every several coherence intervals. In this direction,  we obtain the DEs of $ \xi_{k}^{*} $ and $ \delta^{*} $ as $ M, N, K $ grow large since these expressions will rely only on the correlation matrices and the path-losses. This procedure, deriving the deterministic SINR $ \delta^{*} $ with OLP, will enable us to obtain the optimal phase shifts and amplitudes based only on large-scale statistics (statistical CSI). For the sake of making the problem analytically tractable, we consider 	a system with a common user channel correlation matrix i.e., $ \bR_{k}=\bR $, $ \forall k $ \cite{Kammoun2020}. \textcolor{black}{Although it is possible to consider different correlation matrices, the extension to this case  results in a  mathematically much more involved problem, and is left for future work. In practice, the considered scenario can be met in networks where users are clustered on the basis of their covariance matrices \cite{Adhikary2013,Papazafeiropoulos2022b}.	}  The derivations below rely on the following assumptions.
	
	\begin{assumption}
		$ M $, $ N $, and $ K $ grow large with a bounded ratio as $ 0<\lim \inf\frac{K}{M}\le \lim \sup \frac{K}{M}< \infty $ and $ 0<\lim \inf\frac{N}{M}\le \lim \sup \frac{N}{M}< \infty $, which is equivalent to the notation $\xrightarrow[ n \rightarrow \infty]{\mbox{a.s.}}$ denotes almost sure convergence as $ n \rightarrow \infty $.
	\end{assumption}
	\begin{assumption}
		The correlation matrix $ \bR_{\mathrm{RIS}}  $ will satisfy $ \lim \sup_{N} \|\bR_{\mathrm{RIS}} \|< \infty.$ Similarly, for $ \bR_{\mathrm{BS}} $.
	\end{assumption}
	
	The following theorem exploits the DE analysis to demonstrate that $ \xi_{k}^{*} $ and $ \delta^{*} $ tend asymptotically to deterministic expressions
	\begin{theorem}\label{TheoremDetTau}
		Based on the Assumptions \textcolor{black}{$ 1,2 $}, the  DE of  $ \delta^{*} $, obeying to $ |\delta^{*}-\bar{\delta}|\xrightarrow[ n \rightarrow \infty]{\mbox{a.s.}}0 $ is obtained as 
		the unique positive solution to the following fixed point equation
		\begin{align}
			\bar{\delta}=\tr\bigl(\bR\bigl(\frac{1}{1+\bar{\delta}}((	\bar{\kappa}+1)K+	\bar{\kappa})\bR+\xi \Id_{M}\bigr)^{-1}\bigr),	
		\end{align}
		where  $ \bR=\tr(\bR_{\mathrm{RIS}} \bPhi_{w_{k}} \tilde{\bR}_{\mathrm{RIS}}  \bPhi_{w_{k}}^{\H})\bR_{\mathrm{BS}} $ and $ \eta=\frac{1}{P_{\mathrm{max}}}\sum_{j=1}^{K}\frac{1}{p_{j}}$.
		Also, we have $ |	\xi_{k}^{*}-\bar{	\xi}_{k}|\xrightarrow[ n \rightarrow \infty]{\mbox{a.s.}}0 $ with 
		\begin{align}
			\bar{	\xi_{k}}=\frac{P_{\mathrm{max}}}{p_{k}}\frac{1}{\sum_{j=1}^{K}\frac{1}{p_{j}}}.
		\end{align}
	\end{theorem}
	\begin{proof}
		Please see Appendix~\ref{Theorem1}.	
	\end{proof}
	
	In the special case of not a  STAR-RIS but of a conventional RIS, no HWIs, only a line-of-sight component between the BS and the STAR-RIS, we obtain the DE in \cite{Kammoun2020}.
	
	\subsection{STAR-RIS design}	
	Under OLP with  HWIs and infinite resolution phase shifters at the STAR-RIS, we formulate the PBM optimization problem as
	\begin{align}
		\!\!	(\mathcal{P}3)~&\max_{  \{\thetv,\betv\}}~~~~~~~ \bar{\delta} (\thetv,\betv)\label{Maximization10} \\
		&~\mathrm{s.t.}~~~~~~~~ (\beta_{n}^{t})^{2}+(\beta_{n}^{r})^{2}=1,  \forall n \\
		&~~~~~~~~~~~~~ \beta_{n}^{t}\ge 0, \beta_{n}^{r}\ge 0,~\forall n \\
		&~~~~~~~~~~~~~ |\theta_{n}^{t}|=|\theta_{n}^{r}|=1, ~\forall n
	\end{align}
	where $\thetv=[(\thetv^{t})^{\T}, (\thetv^{r})^{\T}]^{\T} $ and $\betv=[(\betv^{t})^{\T}, (\betv^{r})^{\T}]^{\T} $. We have  vertically stacked $\thetv^{t}$ and $\thetv^{r}$ into  a single vector $\thetv$, and $\betv^{t}$ and $\betv^{r}$ into  a single vector $\betv$ to result in a compact description. Also, for the sake  of exposition, we define two sets: $ \Theta=\{\thetv\ |\ |\theta_{i}^{t}|=|\theta_{i}^{r}|=1,i=1,2,\ldots N\} $, and $ \mathcal{B}=\{\betv\ |\ (\beta_{i}^{t})^{2}+(\beta_{i}^{r})^{2}=1,\beta_{i}^{t}\geq0,\beta_{i}^{r}\geq0,i=1,2,\ldots N\} $, which correspond  together to the feasible set of $ (\mathcal{P}3)$.
	
	Obviously, the  problem $ (\mathcal{P}3)  $ is non-convex and includes coupling among the optimization variables that consist of the amplitudes and the phase shifts for both transmission and reflection. Given that their projection operators of the sets $\Theta$ and $ \mathcal{B}$ can be obtained in closed-form, we resort to the application of  the projected gradient ascent algorithm (PGAM) \cite[Ch. 2]{Bertsekas1999} for the  simultaneous optimization of  $\thetv$ and $\betv$.
	
	We propose a PGAM consisting of the following iterations
	\begin{subequations}\label{mainiteration}\begin{align}
			\thetv^{n+1}&=P_{\Theta}(\thetv^{n}+\mu_{n}\nabla_{\thetv}\bar{\delta} (\thetv^{n},\betv^{n})),\label{step1} \\ \betv^{n+1}&=P_{\mathcal{B}}(\betv^{n}+{\mu}_{n}\nabla_{\betv}\bar{\delta} (\thetv^{n},\betv^{n})),\label{step2} \end{align}
	\end{subequations}
	where the superscript expresses the number of the iteration. As can be seen,  we start from the current iterate $(\thetv^{n},\betv^{n})$  towards the gradient direction with the aim to increase the objective. Note that  $\mu_n$ describes the step size for both $\thetv$ and $\betv$. Moreover, the operators in \eqref{mainiteration}, $P_{\Theta}(\cdot) $ and $ P_{\mathcal{B}}(\cdot) $ describe the projections onto $ \Theta $ and $ \mathcal{B} $, respectively. 
	
	The convergence of PGAM relies on the  suitable choice of the step size in \eqref{step1} and \eqref{step2}. Normally, we would resort to the Lipschitz constant of the gradient to find the ideal step size, but it is difficult to find it for our problem. Hence, we take advantage of the Armijo-Goldstein backtracking line search that returns the step size at each iteration. For this reason, we define a quadratic approximation of $\bar{\delta} (\thetv,\betv)$ as
	\begin{align}
		&	Q_{\mu}(\thetv, \betv;\bx,\by)=\bar{\delta} (\thetv,\betv)\nn\\
		&+\langle	\nabla_{\thetv}\bar{\delta} (\thetv,\betv),\bx-\thetv\rangle-\frac{1}{\mu}\|\bx-\thetv\|^{2}_{2}\nn\\
		&+\langle\nabla_{\betv}\bar{\delta} (\thetv,\betv),\by-\betv\rangle-\frac{1}{\mu}\|\by-\betv\|^{2}_{2}.
	\end{align}    
	The step size $  \mu_{n} $ in \eqref{mainiteration} can be obtained as $ \mu_{n} = L_{n}\kappa^{m_{n}} $, where $ m_{n} $ is the
	smallest nonnegative integer satisfying
	\begin{align}
		\bar{\delta} (\thetv^{n+1},\betv^{n+1})\geq	Q_{L_{n}\kappa^{m_{n}}}(\thetv^{n}, \betv^{n};\thetv^{n+1},\thetv^{n+1}),
	\end{align}
	which can be done by an iterative procedure. Note that $ L_n>0 $, and $ \kappa \in (0,1) $.  The proposed PGAM is summarized in Algorithm \ref{Algoa1}. 
	\begin{algorithm}[th]
		\caption{Projected Gradient Ascent Algorithm for the STAR-RIS Design\label{Algoa1}}
		\begin{algorithmic}[1]
			\STATE Input: $\thetv^{0},\betv^{0},\mu_{1}>0$, $\kappa\in(0,1)$
			\STATE $n\gets1$
			\REPEAT
			\REPEAT
			\STATE $\thetv^{n+1}=P_{\Theta}(\thetv^{n}+\mu_{n}\nabla_{\thetv}\bar{\delta} (\thetv^{n},\betv^{n}))$
			\STATE $\betv^{n+1}=P_{B}(\betv^{n}+\mu_{n}\nabla_{\betv}\bar{\delta} (\thetv^{n},\betv^{n}))$
			\IF{ $\bar{\delta} (\thetv^{n+1},\betv^{n+1})\leq Q_{\mu_{n}}(\thetv^{n},\betv^{n};\thetv^{n+1},\thetv^{n+1})$}
			\STATE $\mu_{n}=\mu_{n}\kappa$
			\ENDIF
			\UNTIL{ $\bar{\delta} (\thetv^{n+1},\betv^{n+1})>Q_{\mu_{n}}(\thetv^{n},\betv^{n};\thetv^{n+1},\thetv^{n+1})$}
			\STATE $\mu_{n+1}\leftarrow\mu_{n}$
			\STATE $n\leftarrow n+1$
			\UNTIL{ convergence}
			\STATE Output: $\thetv^{n+1},\betv^{n+1}$
		\end{algorithmic}
	\end{algorithm} 
	
	Regarding the projection onto the sets  $ \Theta $ and $ \mathcal{B} $, the former, for a given $\thetv\in \mathbb{C}^{2N\times 1}$  $P_{\Theta}(\thetv)$, is given by 
	\begin{equation}
		P_{\Theta}(\thetv)=\thetv/|\thetv|=e^{j\angle\thetv},
	\end{equation}
	where the operations on the right-hand side of the above equation are performed entrywise. 
	
	The latter projection, i.e.,  $P_{ \mathcal{B} }(\betv)$ requires special attention because the constraint $(\beta_{i}^{t})^{2}+(\beta_{i}^{r})^{2}=1,\beta_{i}^{t}\geq0,\beta_{i}^{r}\geq0$  defines the first quadrant of the  unit circle, which makes the expression of the projection onto $\mathcal{B}$ being rather complicated. The projection  $P_{\mathcal{B} }(\betv)$ can become more efficient by allowing  $\beta_{i}^{t}$ and $\beta_{i}^{t}$ to take negative value during the iterative process. We would like to emphasize that this assumption does not affect the optimality of the proposed solution  because we can modify the sign of both $\beta_i^{u}$ and $\theta_{i}^{u}$, $u\in{t,r}$  but still achieve the same objective. Thus,  we can write $P_{ \mathcal{B} }(\betv)$  as
	\begin{subequations}
		\begin{align}
			\left[\ensuremath{P_{\mathcal{B}}(}\boldsymbol{\beta})\right]_{i} & =\frac{\boldsymbol{\beta}_{i}}{\sqrt{\boldsymbol{\beta}_{i}^{2}+\beta_{i+N}^{2}}},i=1,2,\ldots,N,\\
			\left[\ensuremath{P_{\mathcal{B}}(}\boldsymbol{\beta})\right]_{i+N} & =\frac{\boldsymbol{\beta}_{i+N}}{\sqrt{\boldsymbol{\beta}_{i}^{2}+\beta_{i+N}^{2}}}, i=1,2,\ldots,N,
		\end{align}
		while we have projected $\beta_{i}^{t}$ and $\beta_{i}^{t}$ onto the entire unit circle.
	\end{subequations}

	We present the complex-valued gradient in the following proposition.
	\begin{proposition}\label{PropositionGradients}
		The complex gradients $ \nabla_{\thetv}\bar{\delta}(\thetv,\betv) $ and  $\nabla_{\betv}\bar{\delta}(\thetv,\betv) $ are given in closed-forms by
		\begin{align}
			\nabla_{\thetv}\bar{\delta}(\thetv,\betv) &=[\nabla_{\thetv^{t}}\bar{\delta}(\thetv,\betv)^{\T}, \nabla_{\thetv^{r}}\bar{\delta}(\thetv,\betv)^{\T}]^{\T},
		\end{align}

		\begin{subequations}
			\begin{align}
				\nabla_{\thetv^{t}}\bar{\delta}(\thetv,\betv)&=\begin{cases}
					\nu\diag\bigl(\mathbf{A}_{t}\herm\diag(\boldsymbol{{\beta}}^{t})\bigr) & w_{k}=t\\
					0 & w_{k}=r
				\end{cases},\label{derivtheta_t}\\
				\nabla_{\thetv^{r}}\bar{\delta}(\thetv,\betv)&=\begin{cases}
					\nu\diag\bigl(\mathbf{A}_{r}\diag(\boldsymbol{{\beta}}^{r})\bigr) & w_{k}=r\\
					0 & w_{k}=t
				\end{cases}\label{derivtheta_r},\\
			\end{align}
		\end{subequations}
		where  $ \nu=										\hat{\beta}\tr\bigl(\mathbf{R}_{\mathrm{BS}}\bT+\frac{1}{1+\bar{\delta}}((	\bar{\kappa}+1)K+	\bar{\kappa})\bR\bT \mathbf{R}_{\mathrm{BS}}\bT\bigr)(1-\frac{1}{(1+\bar{\delta})^{2}} \tr\bR \bT  \bigl(\	(	\bar{\kappa}+1)K\bR+	\bar{\kappa}\bR\bigr)	)   \bT) $ and $\mathbf{A}_{i}=\mathbf{R}_{\mathrm{RIS}}\bPhi_{i}\tilde{\mathbf{R}}_{\mathrm{RIS}}$ for $ i=t,r $.
		Similarly, the  gradient $\nabla_{\betv}\bar{\delta}(\thetv,\betv) $ is given by
		\begin{align}
			\nabla_{\betv}\bar{\delta}(\thetv,\betv) &=[\nabla_{\betv^{t}}\bar{\delta}(\thetv,\betv)^{\T}, \nabla_{\betv^{r}}\bar{\delta}(\thetv,\betv)^{\T}]^{\T}. 
		\end{align}
		\begin{subequations}
			\begin{align}
				\nabla_{\betv^{t}}\bar{\delta}(\thetv,\betv)&=\begin{cases}
					2\nu\Re\bigl\{\diag\bigl(\mathbf{A}_{k}\herm\diag(\btheta^{t})\bigr)\bigr\} & w_{k}=t\\
					0 & w_{k}=r
				\end{cases}\label{derivbeta_t},\\
				\nabla_{\betv^{r}}\bar{\delta}(\thetv,\betv)&=\begin{cases}
					2\nu\Re\bigl\{\diag\bigl(\mathbf{A}_{k}\herm\diag(\btheta^{r})\bigr)\bigr\} & w_{k}=r\\
					0 & w_{k}=t
				\end{cases}\label{derivbeta_r}.
			\end{align}
		\end{subequations}
	\end{proposition}
	
	\begin{proof}
		Please see Appendix~\ref{prop2}.	
	\end{proof}
	\textcolor{black}{	\subsubsection*{Complexity Analysis of Algorithm \ref{Algoa1}}
		Regarding the complexity analysis, we consider the big-O notation, which is appropriate since we assume  large $M$ and $N$ in this work. At each iteration of Algorithm \ref{Algoa1}, we focus  on  the objective and its gradient value.  First, we focus on $ \bR $, which requires $O(N^2+M^2)$ complex multiplications because $\bPhi_{w_{k}}$ is diagonal. Also, $\bPsi_{k}$ requires   $O(M^3)$ to compute it due to the calculation of the involving matrix inversion. Hence, the  complexity to compute $\bar{\delta}^{*}(\thetv,\betv)$ is $O(K(M^3+N^2))$. 	Next, we focus on the complexity of $\nabla_{\thetv}\bar{\delta}^{*}(\thetv,\betv)$. We observe that by following the above analysis, it is easy to see that  the complexity of computing  the gradients for each iteration is $O(K(M^3+N^2))$. Hence, both the objective and its gradient have the same complexity.}
	
	In the case of the MS protocol, the difference is that the amplitude values are forced to be binary, i.e.,  $\betv^{t}_n\in \{ 0, 1\} $ and $\betv^{r}_n\in \{ 0, 1\} $.  The solution to the  binary constraints belongs to the class of binary nonconvex programming, which is generally NP-hard, which means that it is quite difficult. A practical approach is to find a high-performing solution. Thus, we obtain  a simple solution  by rounding off the solution corresponding to this optimization problem to the nearest binary value, and the result is a reasonably good performance as can be seen in the next section.  In future work, we are going to focus on more advanced methods for solving this optimization problem.
	
	\section{Numerical Results}\label{Numerical}
	In this section, we present the numerical results of the minimum SINR in  STAR-RIS-aided systems. Specifically, we depict analytical results and Monte-Carlo (MC) simulations with $ 10^{3} $ independent channel realizations. 
	
	The simulation setup consists of a STAR-RIS deployed with a UPA of $ N=64 $ elements, which aids the communication between a BS equipped with a ULA of $ M =64$ antennas   that serves $ K = 4 $ UEs. The locations of the BS and RIS, i.e., their  $xy-$coordinates are given as $(x_B,~ y_B) = (0,~0)$ and $(x_R,~ y_R)=(50,~ 10)$, respectively, all in meter units.  Also, UEs in $r$ region are located on a straight line between $(x_R-\frac{1}{2}d_0,~y_R-\frac{1}{2}d_0)$ and $(x_R+\frac{1}{2}d_o,~y_R-\frac{1}{2}d_0)$ with equal distances between each two adjacent UEs. Note that we  have used  $d_0 = 20$~m in our simulations. In a similar way, UEs in the $t$ region are located between $(x_R-\frac{1}{2}d_0,~y_R+\frac{1}{2}d_0)$ and $(x_R+\frac{1}{2}d_o,~y_R+\frac{1}{2}d_0)$. The size of each STAR-RIS element is $ d_{\mathrm{H}}\!=\!d_{\mathrm{V}}\!=\!\lambda/4 $.   Moreover, we have considered a distance-based path-loss, such that the channel gain of a given link $j$ is $\tilde \beta_j = A d_j^{-\alpha_{j}}$, where $A=d_{\mathrm{H}}d_{\mathrm{V}}$ is the area of each reflecting element at the STAR-RIS, and $\alpha_{j}$ is the path-loss exponent. Herein, we have assumed that $\alpha_{r}=2.5$ and $\alpha_{t}=3$ since the UEs in the transmission region are expected to suffer from more severe path loss. Similarly, for $ \tilde \beta $ Regarding the correlation matrices $ \bR_{\mathrm{BS}}$ and $\bR_{\mathrm{RIS}} $, they are  computed according to \cite{Hoydis2013} and \cite{Bjoernson2020}, respectively.  In addition, $ \sigma^2=-174+10\log_{10}B_{\mathrm{c}} $, where $B_{\mathrm{c}}=200~\mathrm{kHz}$ is the bandwidth.
	
	\begin{figure}[!h]
		\begin{center}
			\includegraphics[width=0.8\linewidth]{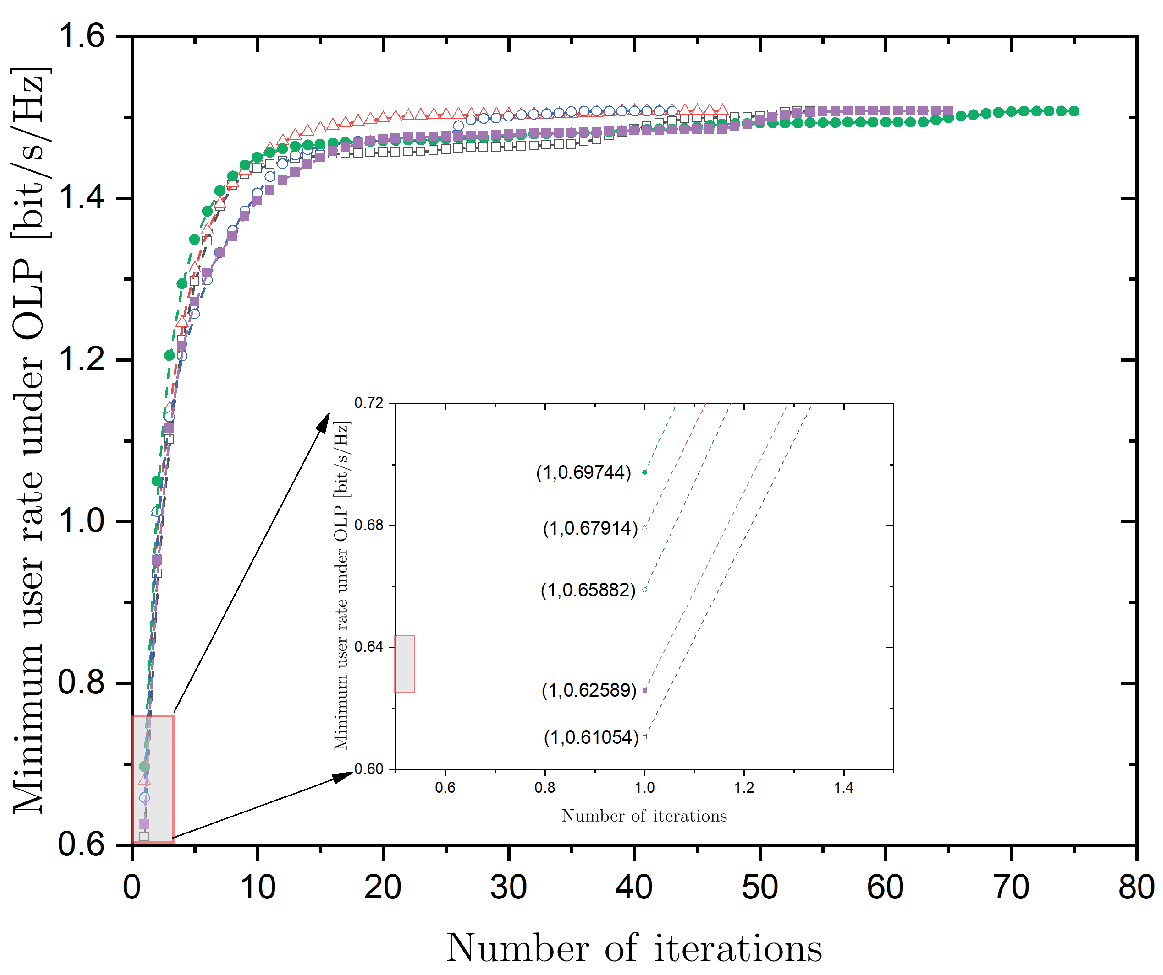}
			\caption{\footnotesize{ Convergence of Algorithm \ref{Algoa1} for an STAR-RIS assisted MIMO system ($M=64$, $ N=64 $, $ K=4 $).  }}
			\label{Fig2}
		\end{center}
	\end{figure}
	
	In Fig. \ref{Fig2}, we show the  convergence of the proposed PGAM. In particular, we depict the minimum user rate against the iteration count obtained by Algorithm \ref{Algoa1} for 5 different randomly generated initial points. The Algorithm \ref{Algoa1} terminates when the  difference of the objective between the two last iterations is less than $10^{-5}$ or the number of iterations is larger than $200$. Given that the optimization problem $ (\mathcal{P}3) $ is nonconvex, the proposed  PGAM can only guarantee a stationary solution, which is not necessarily optimal. As a consequence, Algorithm \ref{Algoa1} may converge to different points starting from different initial points, which is clearly seen in Fig.~\ref{Fig2}. Also, we observe that different initial points result in different convergence rates. For this reason, i.e., for the mitigation of this performance sensitivity of Algorithm  \ref{Algoa1} on the initial points, it is required to execute it from different initial points and select the best convergent solutions. In our extensive simulations, we have followed this observation to allow  a good trade-off between complexity and minimum user rate, i.e., we have executed  Algorithm  \ref{Algoa1} for 5 randomly generated initial points.
	
	\begin{figure}[!h]
		\begin{center}
			\includegraphics[width=0.8\linewidth]{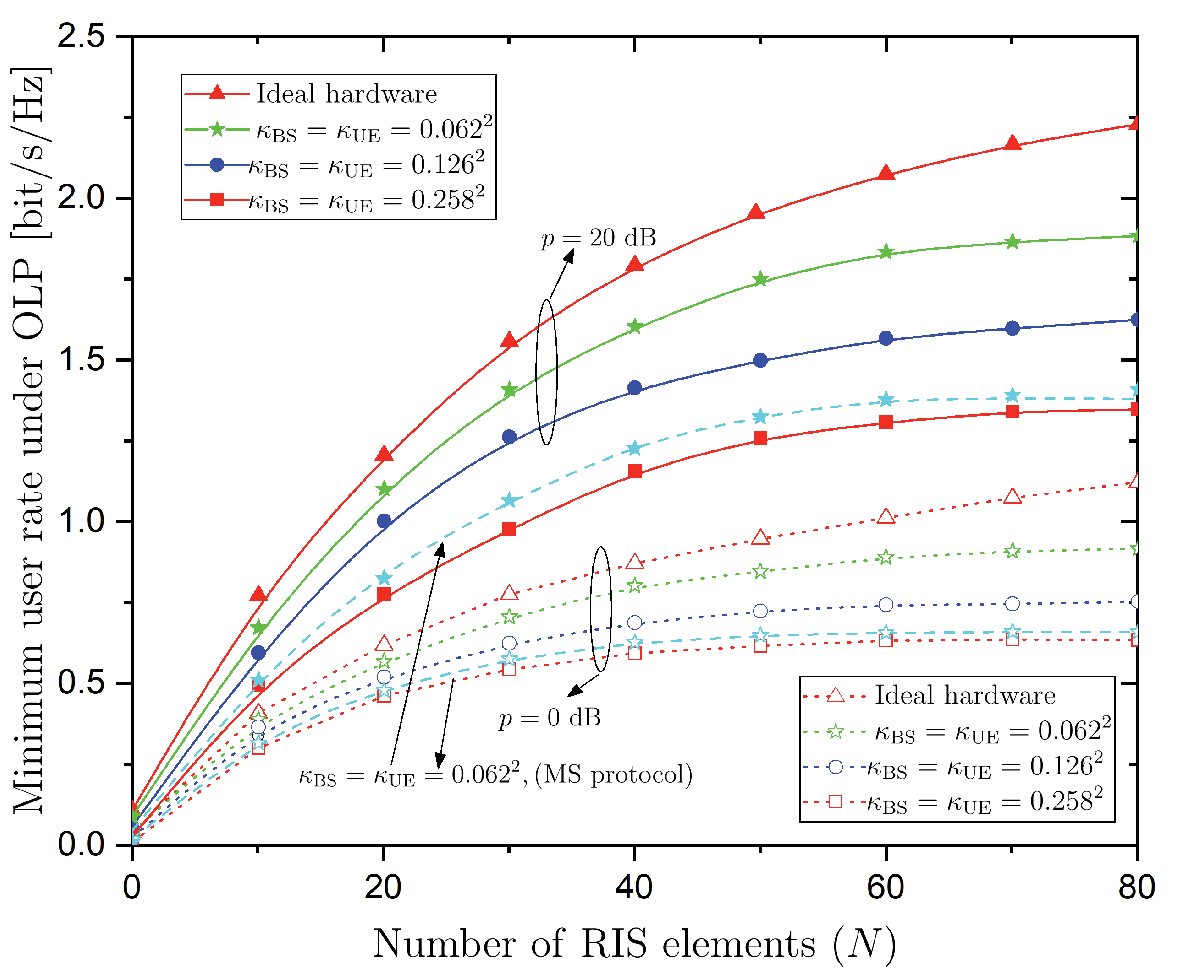}
			\caption{\footnotesize{ Downlink  minimum  user rate   versus the number of RIS elements $N$ of a STAR-RIS  assisted MIMO system with ES  and MS protocols ($ M=16 $, $ K=4 $) for varying T-HWIs $ \kappa_{\mathrm{BS}} $, $\kappa_{\mathrm{UE}}$ and transmit power $ p $.  }}
			\label{Fig3}
		\end{center}
	\end{figure}
	
	In Fig. \ref{Fig3},   we illustrate the minimum rate  versus the  number of elements $ N $ for two distinct SNR values, $ p=0~\mathrm{dB} $ and $ p=20~\mathrm{dB} $. Moreover,  we have shown different values of T-HWIs,  while we have assumed no phase noise. First, we notice that $ \bar{\delta} $, generally, increases with $ N $, but in the case of an imperfect transceiver, it increases in  lower values than in the case of perfect hardware. Also, the higher SNR group presents saturation faster since T-HWIs are power-dependent. Note that  severe T-HWIs,  implying cheaper hardware, result in lower $ \bar{\delta} $. Hence, a trade-off appears between the hardware quality and the performance i.e., the better the hardware quality,  the better the performance. Furthermore,   we observe that the largest percentage of the gain is achieved at a lower number of elements.

	\begin{figure}[!h]
		\begin{center}
			\includegraphics[width=0.8\linewidth]{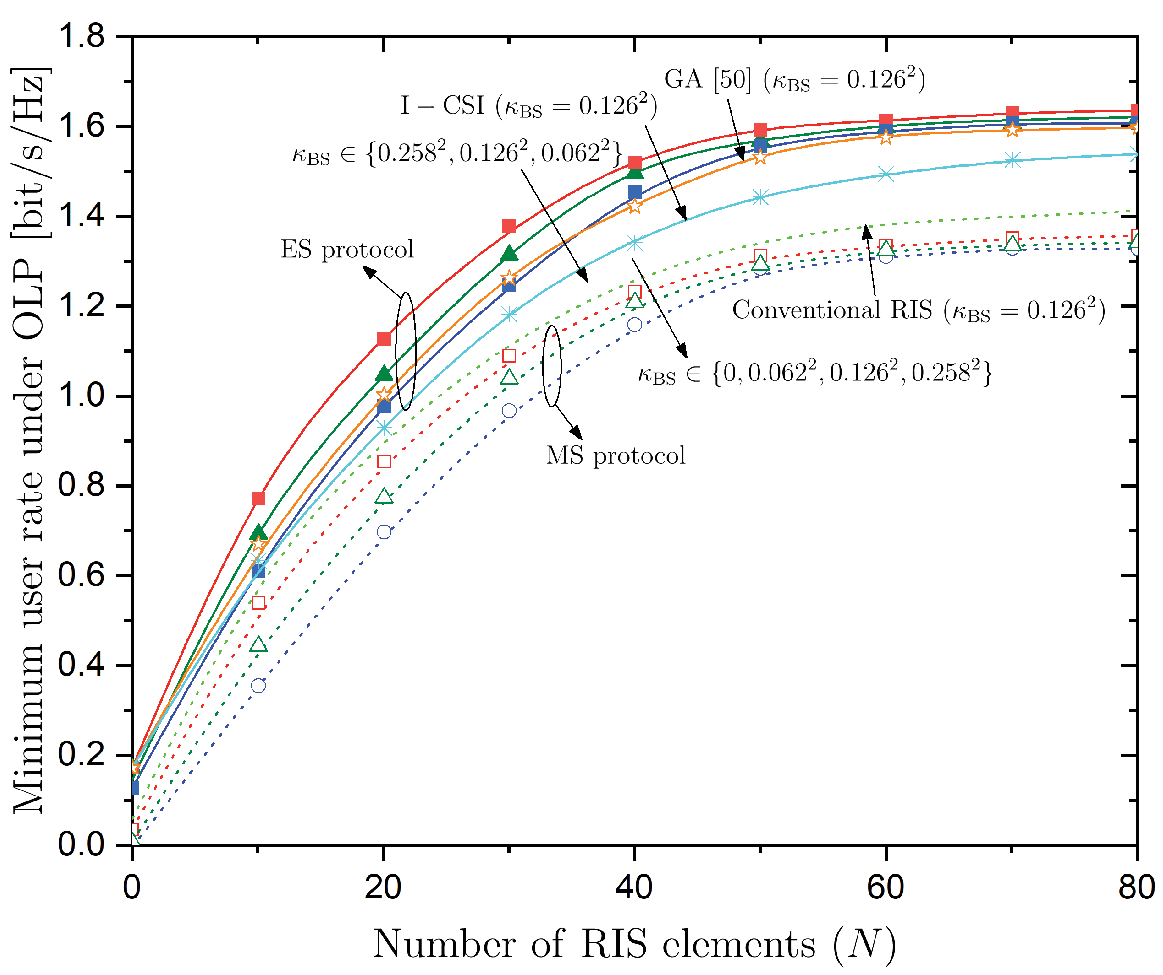}
			\caption{\footnotesize{\textcolor{black}{ Downlink  minimum  user rate  versus the number of RIS elements $N$ of a STAR-RIS  assisted MIMO system ($ M=16 $, $ K=4 $) for varying BS distortion $ \kappa_{\mathrm{BS}} $, both cases of ES and MS protocols, conventional RIS, I-CSI, GA.} }}
			\label{Fig4}
		\end{center}
	\end{figure}
	
	Fig. \ref{Fig4} depicts the minimum rate versus the number of elements $ N $ while we vary only the severity of the  distortion $ \bkappa_{\mathrm{BS}} $, and having  a distortion at the UE $ \bkappa_{\mathrm{UE}} $ set to zero. We observe that the lines converge at a specific value as $ N $ increases. This signifies that the impact of $ \bkappa_{\mathrm{BS}} $ becomes imperceptible when $ N $ is large. In other words, a larger RIS is  preferred since it allows the use of lower quality, i.e., cheaper hardware. Furthermore, in the same figure,  we depict a comparison between the ES and MS protocols.\footnote{\textcolor{black}{In the case of the MS protocol, we have rounded off the solution corresponding to the 		optimization problem to the nearest binary value, and we have found that its impact is minor. Hence, this method is suggested in this case.}} We observe  that, in both cases, the minimum rate increases with the number of RIS elements. The impact of $ \bkappa_{\mathrm{BS}} $ becomes insignificant in the case of the  MS protocol but at larger $ N $. Evidently, in the case of the MS protocol, the minimum rate is slower. \textcolor{black}{ Also, we provide  a comparison with respect to conventional RIS, which presents lower performance compared to STAR-RIS since the latter includes the double number of variables. In addition, for further comparison, we compare the proposed scheme based on S-CSI with full I-CSI. Although I-CSI  yields higher SINR, it results in a lower net sum-rate than  our proposed S-CSI scheme due to the large overhead associated with the full I-CSI acquisition. The figure reveals the superiority of the former. Specifically, we estimate the I-CSI based on the approach in \cite{Nadeem2020}, which has a high training overhead, and maximize the sum-rate expression using the GA  described in \cite{Peng2021} with the same parameters. 			  Also, we compare the proposed optimization, which is the gradient ascent with the GA in \cite{Peng2021}. We observe that the former is preferable because of its better performance and because it provides closed-form expressions.}
	
	\begin{figure}[!h]
		\begin{center}
			\includegraphics[width=0.8\linewidth]{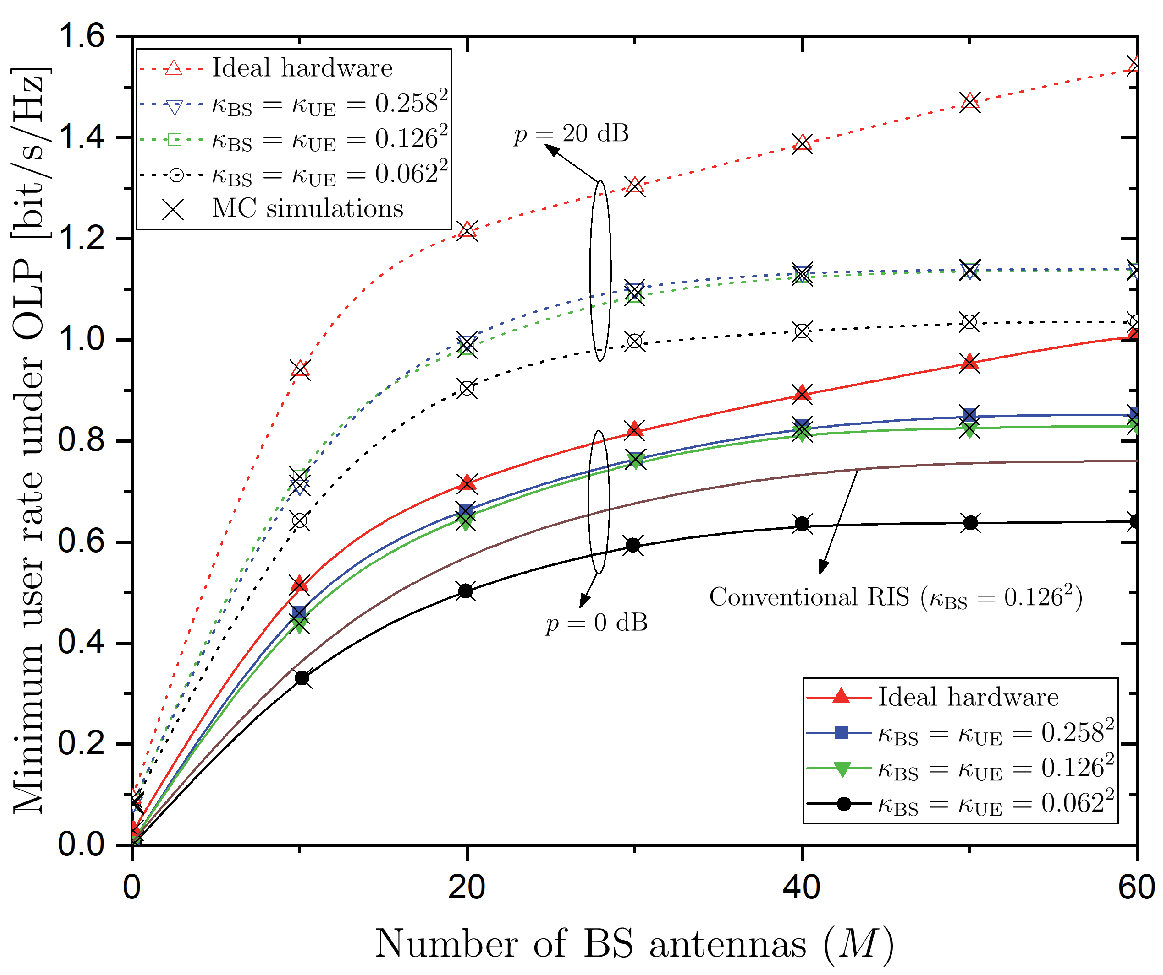}
			\caption{\footnotesize{ Downlink  minimum  user rate  versus the number of BS antennas $M$ of  a STAR-RIS  assisted MIMO system  ($ N=60 $, $ K=4 $) for varying T-HWIs $ \kappa_{\mathrm{BS}} $, $\kappa_{\mathrm{UE}}$ and transmit power $ p $ (Analytical results and MC simulations). }}
			\label{Fig5}
		\end{center}
	\end{figure}
	
	Fig. \ref{Fig5} shows the minimum rate versus the number of BS antennas $ M $ for varying SNR and T-HWI values. As can be seen, this figure resembles Fig. \ref{Fig3}, i.e.,  both figures present a similar dependence on $ M $ and $ N $. The rate appears a ceiling in the case of imperfect hardware where T-HWIs exist, while it grows unbounded as $ M $ increases. We notice a large degradation in the performance when the hardware quality is lower. Moreover, at higher SNR,  the rate is larger but it saturates faster. \textcolor{black}{Note also the superiority of STAR-RIS compared to conventional RIS systems.} In other words, these figures suggest that STAR-RIS-assisted systems perform better at larger values of BS antennas and  RIS  elements. Notably, the analytical results are compared with MC simulations and not only verify the lower bounds but also show their  tightness.
	
	\begin{figure}[!h]
		\begin{center}
			\includegraphics[width=0.8\linewidth]{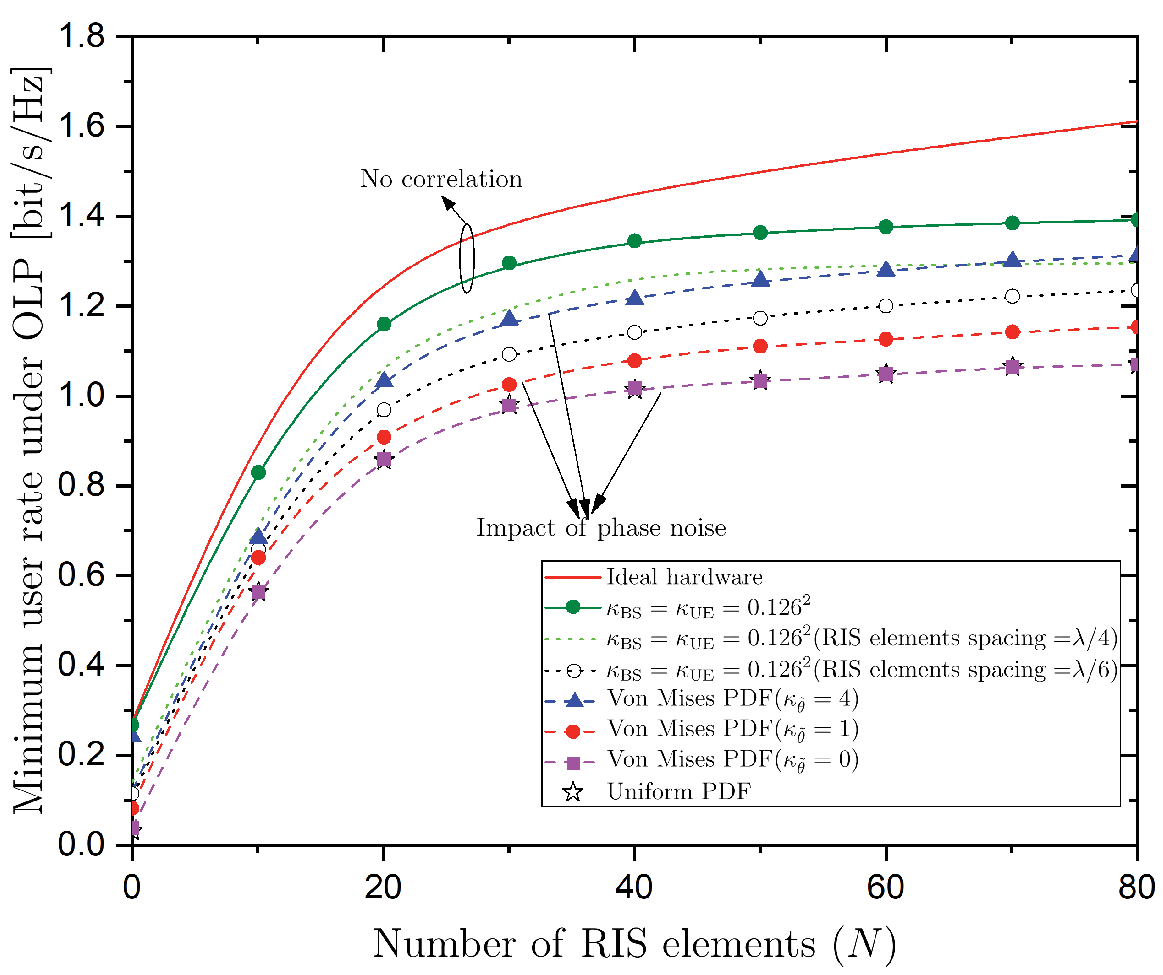}
			\caption{\footnotesize{ Downlink  minimum  user rate  versus the number of RIS elements $N$ of a STAR-RIS assisted MIMO system ($ M=16 $, $ K=4 $) for varying RIS-HWIs and T-HWIs in the cases of correlated/non-correlated Rayleigh fading. }}
			\label{Fig6}
		\end{center}
	\end{figure}
	
	Fig. \ref{Fig6} illustrates the impact of correlated Rayleigh fading on the minimum rate versus the number of RIS elements. In particular, we show the impact of correlation when both perfect and imperfect hardware is implemented. Also, we depict the impact of phase noise. When $ \kappa_{\tilde{\theta}}  $ increases,  the rate  increases, while in the case of uniform phase noise, the rate is the lowest since the phase shifts cannot be optimised because $ \bR_{k} $ does not depend on them. In other words, in the case of the Von Mises distribution, the STAR-RIS can be exploited since the corresponding phase shifts can be optimised. When $ \kappa_{\tilde{\theta}}=0 $, the rate coincides with the line corresponding to the uniform PDF, which verifies our result. Moreover, we observe that the lines corresponding to the phase increase with the number of RIS elements, while the line describing the T-HWI saturates.

	\section{Conclusion}\label{Conclusion}
	\textcolor{black}{	In this paper, we have obtained and studied the minimum SINR with OLP and HWIs while  correlated fading and multiple UEs at each side of the STAR-RIS were taken into account. Next, we have derived the   DEs of the above expressions in the large system limit. Based on the DEs and statistical CSI, we  have presented a low-complexity iterative optimization method maximising the minimum SINR, where the amplitude and the phase shift  were optimised simultaneously at each iteration. As a benchmark, we devised a design based on full I-CSI, and we showed that the proposed S-CSI based scheme is superior in terms of the sum rate due to lower training overhead.  \textcolor{black}{The proposed novel optimization methodology  has a low computational cost, which is very beneficial in STAR-RIS-assisted systems that have not only a large number of elements but also a double number of variables compared to reflecting only RIS.} Moreover, we have evaluated the impact of   correlation and HWIs, and have shed  light on their interplay with  other system parameters toward better network performance.}
	\begin{appendices}
		\section{Useful Lemmas}
		\begin{lemma}[Rank-1 perturbation lemma {\cite[Lemma~2.1]{Bai2}}]
			\\Let $z\in<0$, $\bA\in\CC^{N\times N}$, $\bB\in\CC^{N\times N}$ with $\bB$ Hermitian nonnegative definite, and $\bx\in\CC^{N}$. Then,
			\begin{align}
				|\tr\left((\bB-z\Id_N)^{-1} -(\bB+\bx\bx^\H-z\Id_N)^{-1}\bA|\right)\leq\frac{\|\bA\|}{|z|}.\nn
			\end{align}
		\end{lemma}

		\begin{lemma}\label{lemma:asymptoticLimits}
			Let $\bA \in \bbC^{N \times N}$ with uniformly bounded spectral norm (with respect to $N$). Consider $\bx$, where $\bx \in \bbC^{N}$, $\bx \sim \cC\cN(\b0, \bPhi_{x})$  independent of $\bA$. Then, we have
			\begin{align}
				&\frac{1}{N}\bx^{\H}\bA\bx - \frac{1}{N}\tr \bA\bPhi_{x}  \xrightarrow[ N \rightarrow \infty]{\mbox{a.s.}} 0 \label{eq:oneVector}.			\end{align}
		\end{lemma}
		
		\begin{lemma}{\cite[Lemma~5]{Kammoun2020}}
			We  denote $ \bQ=(\frac{1}{K}\sum_{i=1}^{K}b_{i}\bh_{i}\bh_{i}^{\H}+\sigma^{2}\Id_{M}) $ and $ \bQ_{k}=(\frac{1}{K}\sum_{i\ne k}^{K}b_{i}\bh_{i}\bh_{i}^{\H}+\sigma^{2}\Id_{M}) $. Let $ m(\sigma^{2}) $ be the unique solution to
			\begin{align}
				m(\bA\sigma^{2})=\frac{1}{K}\tr \bA(\frac{1}{K}\sum_{j=1}^{K}\frac{b_{j}\bR_{j}}{1+e_{j}(\sigma^{2})}+\sigma^{2}\Id_{M}),
			\end{align}
			where $ e_{k}(\sigma^{2}) $ is the unique solution to  $ e_{k}(\sigma^{2})=\frac{1}{K} \tr b_{k}\bR_{k}(\frac{1}{K}\sum_{j=1}^{K}\frac{b_{j}\bR_{j}}{1+e_{j}(\sigma^{2})}+\sigma^{2}\Id_{M}) $.
			
			Consider the asymptotic regime outlined in Assumption 3. Let $ [a,b] $ be a closed bounded interval in $ [0,\infty] $, then the
			following convergence holds true 
			\begin{align}
				\sup_{\sigma^{2}\in [a,b]}|\frac{1}{K}\tr\bQ(\sigma^{2}-m(\Id_{M},\sigma^{2}))|\xrightarrow[ N \rightarrow \infty]{\mbox{a.s.}} 0.
			\end{align}
			Using the rank-one perturbation lemma  $ 	\sup_{\sigma^{2}\in [a,b]}|\frac{1}{K}\tr\bQ_{k}(\sigma^{2}-m(\Id_{M},\sigma^{2}))|\xrightarrow[ N \rightarrow \infty]{\mbox{a.s.}} 0. $
		\end{lemma}

		\section{Proof of Theorem~\ref{TheoremOLP}}\label{TheoremOLPproof}
		
		Given that $ C_{ki}=|\bh_{k}^{\H}\bff_{i}|^{2}>0 $ $ \forall k,i $, at the optimal solution to $ 	(\mathcal{P}1) $, all weighted SINRs are equal, i.e., $ \frac{\gamma_{k}\left( \bp, \{\thetv,\betv\}\right)}{a_{k}} $ are equal $ \forall k $. For the sake of exposition, we define the normalized priority vector $ \hat{\ba}\in \mathbb{R}^{K\times 1}_{++} $ and the cross-channel interference matrix $ \bZ\in \mathbb{R}_{+}^{K \times K} $ as
		\begin{align}
			\hat{\ba}&=(\frac{a_{1}}{C_{11}}, \ldots, \frac{a_{K}}{C_{KK}})^{\T},\\
			Z_{ki}&=\left\{
			\begin{array}{lll}
				(					\kappa_{\mathrm{BS}}^{2}+\kappa_{\mathrm{UE}}^{2})\bh_{k}^{\H} \bh_{k},&\mathrm{if}&k= i\\
				C_{ki}+	\bar{\kappa}\bh_{k}^{\H} \bh_{k},&\mathrm{if}&k\ne i
			\end{array}.
			\right. 
		\end{align}
		
		Thus, the SINR can be written as
		\begin{align}
			\gamma_{k}&=\frac{\frac{p_{k}}{C_{kk}}}{\sum_{i\ne k}{p_{i}}Z_{ki}+{p_{k}}Z_{kk} +1},
		\end{align}
		and the weighted SINR is given by
		\begin{align}
			\frac{	\gamma_{k}}{a_{k}}	&=\frac{{p_{k}}}{(\diag({\hat{\ba}})(\bZ\bp+\one))_{k}}.\label{gamma1}
		\end{align}

		Based on the nonlinear Perron–Frobenius theory and by following the steps in \cite{Cai2011}, we are going to propose an algorithm solving $ 	(\mathcal{P}1) $ with $ 	\frac{	\gamma_{k}}{a_{k}} $ given by \eqref{gamma1}. First, having fixed the active beamforming, we focus on the following problem with respect to the power
		\begin{align}
			\!\!	(\mathcal{P}2.1)~&\max_{ \bp } \min_{k} \frac{{p_{k}}}{(\diag({\hat{\ba}})(\bZ\bp+\one))_{k}}\label{Maximization100} \\
			&~\mathrm{s.t.}~~~~~~~~\bw^{\T}\bp \le P_{\mathrm{max}}, ~p_{k}>0, \forall k.\label{Maximization300} 
		\end{align}
		For its solution, we reformulate \eqref{Maximization100} to a geometric program
		by introducing an auxiliary variable $ \delta. $ After making a logarithmic change of variables, i.e., $ \tilde{\delta}=\log\delta $ and $ \tilde{p}_{k}=\log p_{k} $, we obtain the following equivalent convex problem
		\begin{align}
			\!\!	(\mathcal{P}2.2)~ &~~~~~~~~~-\tilde{\delta}\label{Maximization1000} \\
			&~\mathrm{s.t.}~~~~~\min_{\tilde{\delta},\bp} \frac{e^{\tilde{\delta}}(\diag({\hat{\ba}})(\bZ\bp+\one))_{k}}{(e^{\tilde{p}})_{k}}\le 0\\
			&~~~~~~~~~~\frac{1}{P_{\mathrm{max}}}\bw^{\T}e^{\tilde{\bp}}\le 0, \forall k.\label{Maximization3000} 
		\end{align}
		
		The Lagrangian, corresponding to the previous problem is written as
		\begin{align}
			&\mathcal{L}(\tilde{\delta},\tilde{\bp},\bm \lambda, \mu)=-\tilde{\delta}\\
			&+\sum_{k}\lambda_{k}\log(\!\!\frac{e^{\tilde{\delta}}(\diag({\hat{\ba}})(\bZ\bp+\one))_{k}}{(e^{\tilde{p}})_{k}}\!\!)\!+\!\mu \log (\!\frac{1}{P_{\mathrm{max}}}\bw^{\T}e^{\tilde{\bp}}\!),
		\end{align}
		where $ \lambda_{k}$ and $ \mu  $ describe nonnegative Lagrange dual variables.
		
		As can be seen, the problem $ 	(\mathcal{P}2.2) $ is convex while satisfying Slater’s condition. As a result, the Karush–Kuhn–Tucker (KKT) conditions are necessary and sufficient  for the optimality of $ 	(\mathcal{P}2.2) $. We take into account that at the optimal solution of $ 	(\mathcal{P}2.2) $, we can substitute the   inequality constraints with equality. After changing the variables back into $ \delta $ and $ \bp $, we result in the necessary
		and sufficient conditions for optimality of $ (\mathcal{P}2.1) $ given by
		\begin{align}
			\delta^{*}&=\frac{p_{k}^{*}}{(\diag({\hat{\ba}})(\bZ\bp^{*}+\one))_{k}}, \forall k\\
			\bw^{\T}\bp^{*} &=P_{\mathrm{max}}\\
			\delta^{*}&=\frac{x_{k}^{*}}{(\diag({\hat{\ba}})(\bZ\bxi^{*}+\bw))_{k}}, \forall k\label{al1}\\
			\one^{\T}\bm \lambda^{*} &=1\\
			\bxi^{*}&=\frac{\delta^{*}P_{\mathrm{max}}}{\mu^{*}}(\frac{\lambda_{1}^{*}\hat{a}_{1}}{p^{*}_{1}}, \ldots, \frac{\lambda_{K}^{*}\hat{a}_{K}}{p^{*}_{K}})\\
			&\lambda_{K}^{*}>0, \forall k\label{al2}\\
			&	\mu^{*}>0,\label{al3}
		\end{align}
		where $ 	\delta^{*}, \bp^{*} $, and $\bm\lambda^{*}, 	\mu^{*} $ are the optimal primal and dual variables, respectively. Note that \eqref{al1} was derived from $ \pdv{\mathcal{L}}{p_{k}} $ while \eqref{al2}, \eqref{al3} result  because they should be strictly positive to bound $ \tilde{p}_{k}^{*} $ and $ \mathcal{L}(\tilde{\delta}^{*},\tilde{\bp}^{*}, \bm \lambda^{*},\mu^{*}) $ from below.
		
		The  optimal transmit beamformers under the constraint $ \|\bff_{k}\|=1 $ is obtained by rewritting \eqref{al1} as
		\begin{align}
			\xi_{k}^{*}\bff_{k}^{\H}\bh_{k}\bh_{k}^{\H}\bff_{k}&=a_{k}	\delta^{*}\bff_{k}^{\H}(\sum_{i \ne k}\xi_{i}^{*}(	\kappa_{\mathrm{BS}}^{2}+\kappa_{\mathrm{UE}}^{2}+1)\bh_{i}\bh_{i}^{\H}\nn\\&
			+\xi_{k}^{*}\bar{\kappa}\bh_{k}\bh_{k}^{\H}+w_{k}\Id_{M})\bff_{k}\label{paren}.
		\end{align}
		Given that the term in the parentheses in \eqref{paren} is positive definite for any $ \bxi^{*} $, it is invertible. Hence, based on \cite{Cai2011}, the normalized minimum variance distortionless response (MVDR) beamformers in the dual uplink network have the form \eqref{precoder1}. Also, by eliminating the precoder from \eqref{paren}, we obtain the fixed-point equation given by \eqref{xk}. Let $ \bp^{*} $ be the optimal downlink power, then since a decoupling appears in the downlink network, we result in that $ \bp^{*} $ and $ \delta^{*} $ satisfy the fixed-point
		equation 
		\begin{align}
			\frac{1}{\delta}	\bp^{*}=\frac{\delta^{*}}{\bh_{k}^{\H}\bSigma\bh_{k}},
		\end{align}
		where $ \bSigma= (\sum_{i \ne k}\xi_{i}^{*}			(	\kappa_{\mathrm{BS}}^{2}+\kappa_{\mathrm{UE}}^{2}+1)\bh_{i}\bh_{i}^{\H}+\xi_{k}^{*}\bar{\kappa}\bh_{k}\bh_{k}^{\H}+w_{k}\Id_{M})^{-1}$. Now, we rely on the fact that according to the uplink-downlink duality theory, if the downlink and dual uplink networks have the same SINR values $ \forall k $ and $ \bw^{\T}\bp=\one^{\T}\bq $, we obtain \eqref{dlpower}, where $ \bD_{kk}=	\gamma_{k}/|\bh_{k}^{\H}\bff_{k}|^{2} $, 	and we conclude the proof.
		
		\section{Proof of Theorem~\ref{TheoremDetTau}}\label{Theorem1}
		Herein, we are going to find the DEs of $ \xi_{k}^{*} $ and $ \delta^{*} $. First, we denote $ \tilde{\bh}_{k} =p_{k}^{-1/2}\bh_{k}$. Now, $ 	\xi_{k}^{*} $ becomes
		\begin{align}
			\xi_{k}^{*}=\frac{\delta^{*}}{p_{k}\tilde{\bh}_{k}^{\H}\bSigma\tilde{\bh}_{k}}\label{xk1}.
		\end{align}
		
		A direct application of standard DE tools including Lemma 2 is not correct since both $ d_{k} $ and $ \delta^{*} $ depend on the channel vectors $ \tilde{\bh}_{k} $. Note that $ \bSigma $ is now 
		written as
		$ \bSigma= (\sum_{i \ne k}\frac{\delta^{*}}{d_{i}}			(	\bar{\kappa}+1)\tilde{\bh}_{i}\tilde{\bh}_{i}^{\H}+\frac{\delta^{*}}{d_{k}}\bar{\kappa}\tilde{\bh}_{k}\tilde{\bh}_{k}^{\H}+w_{k}\Id_{M})^{-1} $, where $\bar{\kappa}= \kappa_{\mathrm{BS}}^{2}+\kappa_{\mathrm{UE}}^{2} $. 
		
		To continue, we observe that according to the rank-one perturbation lemma, all $ d_{k}=\tilde{\bh}_{k}^{\H}\bSigma\tilde{\bh}_{k}, \forall k$  converge to the same limit in the large system regime. According to the rank-one perturbation lemma, all $ d_{k}=\tilde{\bh}_{k}^{\H}\bSigma\tilde{\bh}_{k}, \forall k$  converge to the same limit in the large system regime. We define $ e_{k}=\frac{d_{k}}{\tilde{d}} $, and assume that $ e_1<\cdots< e_{K} $,
		where $ \tilde{d} $ is obtained as the unique solution to
		\begin{align}
			\tilde{d} =\tr \bigl(\bR\bigl(\frac{\delta^{*}}{\tilde{d}(1+\delta^{*})}(	\bar{\kappa}+1)K\bR+\bar{\kappa}\bR+\Id_{M}\bigr)^{-1}\bigr)\label{tilded}
		\end{align}
		with $ \bR=\tr(\bR_{\mathrm{RIS}} \bPhi_{w_{k}} \tilde{\bR}_{\mathrm{RIS}}  \bPhi_{w_{k}}^{\H})\bR_{\mathrm{BS}} $. Note that we have assumed that all UEs have the same correlation matrix, i.e., $ \bR_{k}=\bR$, $ \forall k $. Also, $	\tilde{d}  $ has been obtained by assuming that all $  d_{k}$ are replaced by the same value $ 	\tilde{d} $ and by  having applied  Lemma 2 and Lemma 3, in order to satisfy
		\begin{align}
			\max_{k}|	d_{k}/\tilde{d}-1|\to 0.\label{dk}
		\end{align}
		
		Also, we define $ \tilde{\bh}_{k} =\bR^{1/2}\bz_{k}$, where $ \bz_{k}\sim \mathcal{CN}(\b0, \Id_{M}) $. Then, we have
		\begin{align}
			d_{k}=\bz_{k}^{\H}\bR^{1/2}\bSigma\bR^{1/2}\bz_{k},\label{dk10}
		\end{align} where 
		$ \bSigma= (\sum_{i \ne k}\frac{\delta^{*}}{e_{i}\tilde{d}}			(	\bar{\kappa}+1)\bR^{1/2}\bz_{i}\bz_{i}^{\H}\bR^{1/2}+\frac{\delta^{*}}{e_{k}\tilde{d}}\bar{\kappa}\bR^{1/2}\bz_{k}\bz_{k}^{\H}\bR^{1/2}+w_{k}\Id_{M})^{-1} $. We divide both sides by $ \tilde{d} $, and write $ d_{k} $ for $ k=K $. We obtain
		\begin{align}
			&	e_{K}\le \bz_{k}^{\H}\bR^{1/2}\bigl(\sum_{i \ne K}\frac{\delta^{*}}{e_{K}}			(	\bar{\kappa}+1)\bR^{1/2}\bz_{i}\bz_{i}^{\H}\bR^{1/2}\nn\\
			&+\frac{\delta^{*}}{e_{K}}\bar{\kappa}\bR^{1/2}\bz_{K}\bz_{K}^{\H}\bR^{1/2}+\tilde{d}w_{K}\Id_{M}\bigr)^{-1}\bR^{1/2}\bz_{k}.
		\end{align}
		
		The proof of $ \lim \sup e_{K}\le 1 $ will be given by using contradiction. In particular, let $ l>0 $ such that $ \lim \sup e_{K}> 1+l $. Then, we can write
		\begin{align}
			&e_{K}\le \bz_{k}^{\H}\bR^{1/2}\bigl(\sum_{i \ne K}\frac{\delta^{*}}{e_{K}}			(	\bar{\kappa}+1)\bR^{1/2}\bz_{i}\bz_{i}^{\H}\bR^{1/2}\nn\\
			&+\frac{\delta^{*}}{e_{K}}\bar{\kappa}\bR^{1/2}\bz_{K}\bz_{K}^{\H}\bR^{1/2}+\tilde{d}(1+l)w_{K}\Id_{M}\bigr)^{-1}\bR^{1/2}\bz_{k},
		\end{align}
		and check that $ \frac{\tilde{d}(1+l) }{\delta^{*}}$ remains almost surely in a bounded interval since we note from \eqref{tilded} that $ \tilde{d}\le \frac{MR }{K} $, where $ R=\lim \sup\|\bR\| $. Hence, the application of Lemmas 2 and 3 results in
		\begin{align}
			&\bz_{k}^{\H}\bR^{1/2}\bigl(\sum_{i \ne K}\frac{\delta^{*}}{e_{K}}			(	\bar{\kappa}+1)\bR^{1/2}\bz_{i}\bz_{i}^{\H}\bR^{1/2}\nn\\
			&+\frac{\delta^{*}}{e_{K}}\bar{\kappa}\bR^{1/2}\bz_{i}\bz_{i}^{\H}\bR^{1/2}+\tilde{d}w_{K}\Id_{M}\bigr)^{-1}\bR^{1/2}\bz_{k}-\mu \xrightarrow[ n \rightarrow \infty]{\mbox{a.s.}}0, 
		\end{align}
		where $ \mu $ can be written as
		\begin{align}
			1=\tr\bigl(\bR\bigl(\mu\frac{\delta^{*}}{1\!+\!\mu \delta^{*}}(	\bar{\kappa}\!+\!1)K\bR\!+\!\bar{\kappa}\bR\!+\!\mu\tilde{d} (1\!+\!l)w_{k}\Id_{M}\bigr)^{-1}\bigr),\label{mu1}
		\end{align}
		with the right hand being a decreasing function of $ \mu $. Based on the definition of $ \tilde{d} $, we can write
		\begin{align}
			1=\tr\bigl(\bR\bigl(\frac{\delta^{*}}{1+ \delta^{*}}(	\bar{\kappa}+1)K\bR+\bar{\kappa}\bR+\mu\tilde{d}w_{k}\Id_{M} \bigr)^{-1}\bigr).\label{mu2}
		\end{align}
		
		Since \eqref{mu1} and  \eqref{mu2} are equal to $ 1 $, and because $ (1+l)>1 $
		and the right hand of \eqref{mu1} is decreasing with $ \mu $, it is  required that $ \mu<1 $. However, a contradiction appears when the size of the system increases because $ \mu \ge 1 $ from \eqref{mu1}. As a result, we have shown that $ \lim \sup e_{K}\le 1  $. Similarly, we can show that $  \lim \inf e_{1}\ge 1 $. Hence, we have proved \eqref{dk}.
		
		If we use \eqref{dk} in the definition of $ \delta^{*} $, we obtain
		\begin{align}
			\delta^{*}=\frac{P_{\mathrm{max}}}{\sum_{k=1}^{K}\frac{1}{\tilde{d}p_{k}}}+o(1). \end{align}
		If we replace $ \tilde{d} $ by $ \frac{\delta^{*}\sum_{k=1}^{K}\frac{1}{p_{k}}}{P_{\mathrm{max}}} $,
		and use  \eqref{tilded}, we obtain
		\begin{align}
			\delta^{*} =\tr(\bR(\frac{1}{1+ \delta^{*}}(	\bar{\kappa}+1)K\bR+\bar{\kappa}\bR+\eta w_{k}\Id_{M} )^{-1})+o(1),
		\end{align}
		where $ \eta=\frac{1}{P_{\mathrm{max}}} \sum_{k=1}^{K}\frac{1}{p_{k}}$. By discarding the vanishing terms,  the DE of $ \delta^{*} $ is given by $ \bar{\delta} $, which is the unique solution to
		\begin{align}
			\bar{\delta}=\tr\bigl(\bR\bigl(\frac{1}{1+ \delta^{*}}(	\bar{\kappa}+1)K\bR+\bar{\kappa}\bR+\eta w_{k}\Id_{M} \bigr)^{-1}\bigr).
		\end{align}
		Having shown the convergence of $\delta^{*}  $ and 	$ d_{k} $ together in $ \eqref{dk10} $, the convergence of $ \xi_{k}, \forall k $ is shown, and the proof is concluded.
		
		\section{Proof of Proposition~\ref{PropositionGradients}}\label{prop2}
		First, we obtain $\nabla_{\boldsymbol{\theta^{t}}}\bar{\delta}(\thetv,\betv) $, which is the complex gradient of the SINR with respect to $ \boldsymbol{\theta}^{t\ast}$. We observe from \eqref{cov1} that $\nabla_{\thetv^{t}}\bar{\delta}(\thetv,\betv) = 0$ if $w_k=r$, i.e., in the case that UE $k$ is in the reflection region. Hence, we focus on the derivation of $\nabla_{\thetv^{t}}\bar{\delta}(\thetv,\betv)$ when $w_k=t$. For the sake of exposition, we  write $\bR$ that is going to be used below as
		\begin{align}
			\bR&=\hat{\beta}_{k}\tr(\bR_{\mathrm{RIS}} \bPhi_{t} \tilde{\bR}_{\mathrm{RIS}}  \bPhi_{t}^{\H})\bR_{\mathrm{BS}}=\hat{\beta}_{k}\tr(\bA_{t}\bPhi_{t}^{\H}  )\bR_{\mathrm{BS}}\label{cov1t}.
		\end{align}
		where  $\mathbf{A}_{t}=\mathbf{R}_{\mathrm{RIS}}\bPhi_{t}\tilde{\mathbf{R}}_{\mathrm{RIS}}$. Similarly, if  $w_k=r$, we define $\mathbf{A}_{r}=\mathbf{R}_{\mathrm{RIS}}\bPhi_{r}\tilde{\mathbf{R}}_{\mathrm{RIS}}$.

		Next, we   define the implicit function  
		\begin{align}
			g(\bar{\delta}, \bPhi_{t})=\bar{\delta}-	\tr\bR_{k} \bT,
		\end{align}
		where $ \bT=\bigl(\frac{1}{1+\bar{\delta}}	(	\bar{\kappa}+1)K\bR+\bar{\kappa}\bR+\eta  \Id_{M}\bigr)^{-1} $.
		We exploit the implicit function theorem that indicates
		\begin{align}
			\nabla_{\thetv^{t}}\bar{\delta}=-\nabla_{\thetv^{t}}g/\nabla_{\bar{\delta}}g,
		\end{align}
		where $ \nabla_{\bar{\delta}}g=1- \tr\bR \bT \bigl(\frac{1}{(1+\bar{\delta})^{2}}		(	\bar{\kappa}+1)K\bR+\bar{\kappa}\bR\bigr)   \bT$. Note that we have used the property concerning the derivative of an inverse matrix.

		For the calculation of $\nabla_{\thetv^{t}}g$, we follow similar steps to  \cite[Chap. 3]{hjorungnes:2011}. The differential $d(g)$, where $d(\cdot)$ is  obtained as
		\begin{align}
			d(g)	=1-\tr(d(\bR)\bT+\bR  d(\bT))\label{dg}.
		\end{align}
		
		The differential $d(\mathbf{R})$  is derived by noting that 
		\begin{align}
			d(\mathbf{R})\label{dRk:general}  
			& =\hat{\beta}\mathbf{R}_{\mathrm{BS}}\tr\bigl(\mathbf{A}_{t}\herm d(\bPhi_{t})+\mathbf{A}_{t}d\bigl(\bPhi_{t}\herm\bigr)\bigr).
		\end{align}
		
		Given that $\bPhi_{t}$ is diagonal,   $d(\mathbf{R})$ can be written as
		\begin{align}
			d(\mathbf{R})&=\hat{\beta}\mathbf{R}_{\mathrm{BS}}\bigl(\bigl(\diag\bigl(\mathbf{A}_{t}\herm\diag(\boldsymbol{{\beta}}^{t})\bigr)\bigr)\trans d(\boldsymbol{\theta}^{t})\nn\\
			&+\bigl(\diag\bigl(\mathbf{A}_{t}\diag(\boldsymbol{{\beta}}^{t})\bigr)\bigr)\trans d(\boldsymbol{\theta}^{t\ast})\bigr).\label{eq:dRk}
		\end{align}
		
		Regarding the differential of $ \bT $, we have
		\begin{align}
			d(\bT)=\frac{1}{1+\bar{\delta}}\bT 	((	\bar{\kappa}+1)K+\bar{\kappa})d(\bR)\bT\label{dt}.
		\end{align}
		Substitution of \eqref{eq:dRk} and \eqref{dt} in \eqref{dg} results in
		\begin{align}
			\nabla_{\thetv^{t}}g&=\frac{\partial}{\partial{\thetv^{t\ast}}}g\nn\\
			&=-\hat{\beta}\tr\bigl(\mathbf{R}_{\mathrm{BS}}\bT+\frac{1}{1+\bar{\delta}}((	\bar{\kappa}+1)K+\bar{\kappa})\bR\bT \mathbf{R}_{\mathrm{BS}}\bT\bigr)
			\nn\\
			&\times\diag\bigl(\mathbf{A}_{t}\herm\diag(\boldsymbol{{\beta}}^{t})\bigr),
		\end{align}
		which gives \eqref{derivtheta_t}. Proof of \eqref{derivtheta_r} can be easily  obtained by following the same steps. We omit the details for brevity.
		
		Regarding the derivation of   $\nabla_{\boldsymbol{\beta}^{t}}\bar{\delta}$, we follow the procedure below. In particular, we start with the derivation of  $\nabla_{\boldsymbol{\beta}^{t}}\bar{\delta}$ when $w_{k}=t$. In  this case, from \eqref{dRk:general}, we obtain
		\begin{subequations}\label{dRk_beta_t}
			\begin{align}
				d(\mathbf{R}) & =\hat{\beta}\mathbf{R}_{\mathrm{BS}}\tr\bigl(\mathbf{A}_{t}\herm d(\bPhi_{t})+\mathbf{A}_{t}d\bigl(\bPhi_{t}\herm\bigr)\bigr)\\
				& =\hat{\beta}\mathbf{R}_{\mathrm{BS}}\bigl(\diag\bigl(\mathbf{A}_{t}\herm\diag(\btheta^{t})\bigr)^{\T}d(\boldsymbol{\beta}^{t})\nn\\
				&+\diag\bigl(\mathbf{A}_{t}\diag(\btheta^{t\ast})\bigr)^{\T}d(\boldsymbol{\beta}^{t})\bigr)\\
				& =2\hat{\beta}\mathbf{R}_{\mathrm{BS}}\Re\bigl\{\diag\bigl(\mathbf{A}_{t}\herm\diag(\btheta^{t})\bigr\}^{\T} d(\boldsymbol{\beta}^{t}).
			\end{align}
		\end{subequations}
		By inserting \eqref{dt}  and \eqref{dRk_beta_t} in \eqref{dg}, we obtain
		\begin{align}
			\nabla_{\boldsymbol{\beta}^{t}}\bar{\delta} 
			&=-2\hat{\beta}\tr\bigl(\mathbf{R}_{\mathrm{BS}}\bT+\frac{1}{1+\bar{\delta}}((	\bar{\kappa}+1)K+\bar{\kappa})\bR_{k}\bT \mathbf{R}_{\mathrm{BS}}\bT\bigr)\nn\\
			&\times\Re\bigl\{\diag\bigl(\mathbf{A}_{t}\herm\diag(\btheta^{t})\bigr\}.
		\end{align}
		Similarly, we can write  $	\nabla_{\boldsymbol{\beta}^{r}}\bar{\delta} $ as 
		\begin{align}
			\nabla_{\boldsymbol{\beta}^{r}}\bar{\delta} 
			&=-2\hat{\beta}\tr\bigl(\mathbf{R}_{\mathrm{BS}}\bT+\frac{1}{1+\bar{\delta}}((	\bar{\kappa}+1)K+\bar{\kappa})\bR_{k}\bT \mathbf{R}_{\mathrm{BS}}\bT\bigr)\nn\\
			&\times\Re\bigl\{\diag\bigl(\mathbf{A}_{r}\herm\diag(\btheta^{r})\bigr\}.
		\end{align}
		Note that $\nabla_{\betv}\bar{\delta}(\thetv,\betv)$ is real-valued.
		
	\end{appendices} 
	
	\bibliographystyle{IEEEtran}
	\bibliography{IEEEabrv,mybib}
	
\end{document}